\documentclass[11pt]{article}
\usepackage{amsmath}
\usepackage{graphicx}
\usepackage{centernot}
\usepackage{natbib}
\usepackage{url} 
\usepackage[shortlabels]{enumitem}
\usepackage{latexsym, epsfig, amssymb, amsthm, mathrsfs, bbm, enumitem}
\usepackage[OT1]{fontenc}
\usepackage{xcolor}
\usepackage[colorlinks,citecolor=blue,linkcolor=blue,urlcolor=blue]{hyperref}
\usepackage{multirow,multicol,booktabs}
\usepackage{anysize}

\marginsize{1.1in}{1.1in}{0.55in}{0.8in} %

\makeatletter

\usepackage{tikz}
\usepackage{float}
\usepackage{booktabs}
\usepackage{adjustbox}

\usepackage{tabularx}
\usepackage{multirow}
\usepackage{algpseudocode}
\usepackage{algorithm}

\newtheorem{condition}{Condition}
\newtheorem{remark}{Remark}

\def\Xi{X^{(i)}}

\newcommand{\ba}{\mbox{\bf a}}

\newcommand{\br}{\mbox{\bf r}}
\newcommand{\bs}{\mbox{\bf s}}
\newcommand{\bbe}{\mbox{\bf e}}
\newcommand{\bu}{\mbox{\bf u}}

\newcommand{\bx}{\mbox{\bf x}}

\newcommand{\bz}{\mbox{\bf z}}
\newcommand{\bA}{\mbox{\bf A}}

\newcommand{\bD}{\mbox{\bf D}}

\newcommand{\bI}{\mbox{\bf I}}

\newcommand{\bQ}{\mbox{\bf Q}}

\newcommand{\bU}{\mbox{\bf U}}
\newcommand{\bV}{\mbox{\bf V}}
\newcommand{\bW}{\mbox{\bf W}}
\newcommand{\bX}{\mbox{\bf X}}

\newcommand{\bY}{\mbox{\bf Y}}
\newcommand{\bZ}{\mbox{\bf Z}}

\newcommand{\bzero}{\mbox{\bf 0}}

\newcommand{\balpha}{\mbox{\boldmath $\alpha$}}
\newcommand{\bbeta}{\mbox{\boldmath $\beta$}}
\newcommand{\brho}{\mbox{\boldmath $\rho$}}
\newcommand{\bnu}{\mbox{\boldmath $\nu$}}

\newcommand{\bgamma}{\mbox{\boldmath $\gamma$}}
\newcommand{\bPsi}{\mbox{\boldmath $\Psi$}}
\newcommand{\bPi}{\mbox{\boldmath $\Pi$}}
\newcommand{\bpsi}{\mbox{\boldmath $\psi$}}
\newcommand{\bphi}{\mbox{\boldmath $\phi$}}
\newcommand{\bet}{\mbox{\boldmath $\eta$}}
\newcommand{\bmeta}{\mbox{\boldmath $\eta$}}

\newcommand{\bmu}{\mbox{\boldmath $\mu$}}
\newcommand{\bzeta}{\mbox{\boldmath $\zeta$}}

\newcommand{\bSigma}{\mbox{\boldmath $\Sigma$}}

\newcommand{\bmathcalV}{\mbox{\boldmath $\mathcal{V}$}}

\newcommand{\var}{\mathrm{var}}

\newcommand{\diag}{\mathrm{diag}}

\newtheorem{theo}{\sc Theorem}

\newtheorem{prop}{Proposition}
\newtheorem{coll}{\sc Corollary}

\renewcommand{\theequation}{{\arabic{section}.\arabic{equation}}}

\begin{document}

\title{Robust Subgroup Analysis for Heterogeneous Censored Data%
\thanks{This work was accepted for publication in \href{https://onlinelibrary.wiley.com/journal/20491573}{\textit{Stat}} on April 30, 2026. 
The final published version is available at \url{https://doi.org/10.1002/sta4.70158}.
Zhaohui Xu is Ph.D. candidate, International Institute of Finance, School of
Management, University of Science and
Technology of China, Hefei, 230026, China
(E-mail: xzh2517@mail.ustc.edu.cn). %
Daoji Li is Associate Professor, Department of Information Systems and Decision Sciences, College of Business and Economics, California State
University, Fullerton, CA, 92831, United States (E-mail: dali@fullerton.edu). %
Zemin Zheng is Professor, International Institute of Finance, School of
Management, University of Science and
Technology of China, Hefei, 230026, China (E-mail: zhengzm@ustc.edu.cn). %
}
\date{}
\date{}
\author{Zhaohui Xu$^1$, Daoji Li$^2$ and Zemin Zheng$^1$
\medskip\\
$^1$California State University, Fullerton
\\ $^2$
University of Science and
Technology of China
\\
} %
}

\maketitle

\begin{abstract}	
Subgroup analysis is important in practice because real-world data typically come from heterogeneous populations, where meaningful patterns can differ substantially across subpopulations. Correctly identifying these subgroups can improve prediction accuracy, prevent biased or misleading conclusions, and support more effective, targeted decision-making.  
While most existing subgroup analysis methods are developed for complete data, in this paper we propose a novel and robust approach for censored data under heterogeneous accelerated failure time (AFT) models.  Specifically, we combine inverse probability weighting, M-estimation, and concave pairwise fusion penalization to simultaneously identify subgroups and estimate covariate effects for heterogeneous censored data, without requiring prior knowledge of individual subgroup memberships. 
We further develop an efficient RISA-ADMM algorithm to implement the method and establish its convergence. Furthermore, we derive the theoretical properties of the proposed estimators under mild regularity conditions. 
Extensive simulations and an application to the German credit dataset demonstrate the robustness and effectiveness of our approach.
\end{abstract}
	
\textit{Running title}: RISA


\newpage
\section{Introduction}\label{sec:Intro}

In many applications, data are inherently heterogeneous, meaning that the effects of covariates can vary across distinct subpopulations. For example, in risk management, the effect of income on default risk may differ by region, reducing risk more in low-cost areas than in high-cost metropolitan areas. 
In public health, the effectiveness of a vaccination program may vary by age group, with children and older adults responding differently to the same vaccine. In marketing, the impact of a promotional campaign may differ across customer segments, with younger consumers more likely to respond to online advertisements than older consumers.  Ignoring such heterogeneity can lead to biased estimates and misleading conclusions~\citep{pei2024latent}.  

Over the past decades, various subgroup analysis methods have been developed.   Broadly speaking, existing approaches for subgroup analysis in the literature can be categorized into two classes. The first class comprises mixture model–based approaches,
including Gaussian mixture models \citep{wei2013latent}, low-rank methods for mixtures of linear regressions \citep{chaganty2013spectral}, and logistic-normal mixture models \citep{shen2015inference}.   
However,  these methods typically require specifying the number of subgroups and assuming a particular data distribution, both of which are often unknown or difficult to justify in practice. The second class consists of fusion penalty–based approaches. For example, 
\cite{ma2017concave} 
proposed a pairwise fusion penalization approach for linear regression models that automatically identifies subgroups and partitions observations accordingly, without requiring prior knowledge of the number of subgroups. The central idea of~\cite{ma2017concave} is to represent subgroup heterogeneity through group-specific intercepts, after accounting for covariate effects, and to enforce subgroup fusion by applying concave penalties, such as SCAD~\citep{fan2001variable}, 
to the pairwise differences among these intercepts.  Following this idea, the 
fusion penalty–based approach in~\cite{ma2017concave} has been extended to different model settings, including Poisson regression \citep{chen2019subgroup}, quantile regression \citep{zhang2019robust, lu2021multiply},  
additive partially linear models \citep{liu2019subgroup, cai2024subgroup}, functional linear regression~\citep{li2021clusterwise, zhang2022subgroup}, functional partially linear regression~\citep{ma2023subgroup, wang2023latent}, panel data models~\citep{wang2022group, wang2024homogeneity}, and multivariate response regression~\citep{wu2026simultaneous}.  However, the aforementioned methods mainly focus on complete data.

To accommodate censored data, \cite{yan2021subgroup} extended the pairwise fusion penalization framework in~\cite{ma2017concave} to accelerated failure time (AFT) models. In particular, they employed imputation to handle censored outcomes and assumed sub-Gaussian errors to establish theoretical guarantees. However, subgroup structures may be distorted due to incorrect imputation. Moreover, in many applications, such as finance and biomedical studies, survival times are often contaminated by 
heavy-tailed noise, which violates the error assumption, substantially degrades estimation accuracy, and compromises subgroup identification. Motivated by these considerations, we propose a new and robust subgroup analysis approach that preserves subgroup-specific structures, is resilient to heavy-tailed errors, 
and provides theoretical guarantees under mild conditions. Specifically, under the AFT model, we handle censoring directly via inverse probability weighting, capture heterogeneity using pairwise fusion penalization, and achieve robust estimation through  
M-estimation.

Besides using different strategies for handling censored data,  our work differs from~\cite{yan2021subgroup} in both methodology and theory, as highlighted in the following key aspects: (a) our method only requires estimating the survival function of the censoring time, whereas \cite{yan2021subgroup} requires estimating the full error distribution; (b) unlike \cite{yan2021subgroup}, which updates estimators using the Buckley–James iterative procedure with the ADMM algorithm, we adopt iterative penalized reweighted least squares (IPRLS);  (c) \cite{yan2021subgroup} imposes a sub-Gaussian assumption on the error distribution, while our method can accommodate heavy-tailed errors without this assumption; and (d) \cite{yan2021subgroup} establishes theoretical results under fixed-dimensional settings, whereas our framework allows the covariate dimensions to grow with the sample size under suitable rate conditions.

Our main contributions are fourfold. First, we propose a robust subgroup analysis approach for heterogeneous censored data that accommodates heavy-tailed errors by leveraging M-estimation.  Our method relaxes the sub-Gaussian assumption on the error distribution in \cite{yan2021subgroup} and is more robust to outliers and heavy-tailed errors when identifying the number of subgroups. 
Second, we directly account for censored observations using inverse probability weighting. Compared with the imputation-based approach in \cite{yan2021subgroup}, our method avoids potential distortion of subgroup structures due to incorrect imputation and is easier to implement, as it only requires estimating the survival function of the censoring time rather than the full error distribution. Third, our method is highly flexible, employing iterative penalized reweighted least squares (IPRLS) to update the estimators in each iteration. This implementation differs substantially from those in~\cite{ma2017concave} and \cite{yan2021subgroup}. While we illustrate the method using Huber loss, the algorithm can be readily extended to other loss functions.  Fourth, we extend the fusion-penalty framework for subgroup identification to censored data in the diverging-dimensional setting, where the covariate dimensions $p$ and $q$ can grow with the sample size $n$. In contrast, \cite{yan2021subgroup} establishes asymptotic results only under fixed-dimensional settings. Our method substantially broadens the applicability of the fusion-penalty framework to a broader range of practical scenarios.

The rest of this paper is organized as follows. Section~\ref{section2} introduces the problem setup and the proposed methodology. Section~\ref{section3} presents a new algorithm for solving the resulting optimization problem. Section~\ref{section4} establishes the theoretical properties of the proposed estimators. Sections~\ref{section5} and~\ref{section6} report simulation studies and a real data analysis, respectively. Section~\ref{section7} concludes with discussions.  
All proofs and technical details are provided in the Supplementary Material.

{\bf Notations.} For a vector $\bx=(x_1,\cdots,x_p)^{\top}$, 
let $\|\bx\|_0 = |\{j: x_j \neq 0\}|$, $\|\bx\|_1=\sum_{j=1}^p|x_j|$, $\|\bx\|_2=\|\bx\|=\sqrt{\sum_{j=1}^p|x_j|^2}$, and $\|\bx\|_{\infty}=\max_{1\le i\le p}|x_i|$. For a matrix $\bA=(a_{ij})\in \mathbb{R}^{p\times q}$, 
let $\|\bA\|_1=\max_{1\le j\le q}\sum_{i=1}^p|a_{ij}|$,  $\|\bA\|_2=\|\bA\|= \Lambda_{\max}(\bA^{\top}\bA)^{1/2}$, and $\|\bA\|_{\infty}=\max_{1\le i\le p}\sum_{j=1}^q|a_{ij}|$, where $\Lambda_{\min}(\bA)$ and $\Lambda_{\max}(\bA)$ are the minimum and  maximum eigenvalues of a given matrix $\bA$, respectively. 
For two positive sequences $\{a_n\}_{n\ge 1}$ and $\{b_n\}_{n\ge 1}$, 
we write $a_n\gg b_n$ if $a_n^{-1}b_n=o(1)$.
For a function $f(t)$, let $\dot{f}(t)$ and $\ddot{f}(t)$ denote the first- and second-order derivatives of $f(t)$, respectively.

\section{Model and Method}\label{section2}
\subsection{Heterogeneous AFT models}

Let $T_i$ denote the failure time of subject $i$. In the context of credit risk, this corresponds to the default time, defined as the period from borrowing to default.
The classical AFT model assumes that
\begin{align}\label{eq: classical-AFT}
	\log (T_i)=  \mu + \bx_i^{\top}\bbeta + \varepsilon_i,\quad \,\, i=1, 2, \cdots, n,
\end{align}
where $\mu$ is the intercept, $\bx_i=(x_{i1}, \cdots, x_{ip})^{\top}\in\mathbb{R}^p$ is the covariate vector for subject $i$, $\bbeta=(\beta_1, \cdots, \beta_p)^{\top}\in\mathbb{R}^p$ is the unknown regression coefficient vector, and $\varepsilon_i$ is the error term independent of  $\bx_i$. 
As shown in~\eqref{eq: classical-AFT}, the classical AFT model assumes that all covariate effects are the same for every subject in the population, which may not hold in practice.
For example, In credit risk management, \cite{pei2024latent} observed that 
some covariate effects may vary across subgroups.

Therefore, we consider the following heterogeneous AFT model 
\begin{align}\label{eq: Heterogeneous-AFT}
	\log (T_i)=  \bz_i^{\top}\balpha_i + \bx_i^{\top}\bbeta + \varepsilon_i,\quad \,\, i=1, 2, \cdots, n,
\end{align}
where $\bz_i$ is a $q$-dimensional covariate vector for the $i$th subject contributing to the heterogeneity after adjusting for the effects of a set of covariates $\bx_i$, $\balpha_i=(\alpha_1,\cdots,\alpha_q)\in\mathbb{R}^q$ is subject-specific effect of $\bz_i$, $\bbeta=(\beta_1, \cdots, \beta_p)^{\top}\in\mathbb{R}^p$ is the covariate effects of $\bx_i$, and $\varepsilon_i$ is the error term independent of $\bz_i$ and $\bx_i$. 
We assume that all subjects are from $K$ different subgroups $G_1, \cdots, G_K$ with $K\geq 1$, where $K$ is much smaller than the sample size $n$, $G=(G_1, \cdots, G_K)$ is a partition of $\{1, \cdots, n\}$, and subjects within the same group share the same covariate effects for $\bz_i$.  In other words, we have $\balpha_i = \brho_k$ for all $i \in G_k$, $k = 1, \dots, K$, where $\brho_k$ is the common covariate effects of $\bz_i$ in subgroup $k$. Both the number of subgroups $K$ and the subgroup structure $G$ are unknown.

In practice, 
the failure time $T_i$ is often subject to right censoring. Consequently, we observe only $\left\{(T_i^{*}, \bx_i^{\top}, \bz_i^{\top}, \Delta_i)\right\}_{i=1}^n$, where $T_i^{*}=\min\{T_i, C_i\}$ is the observed time, $C_i$ is the censoring time, and $\Delta_i=I(T_i\leq C_i)$ is the censoring indicator. Our goal is to simultaneously determine the number of subgroups $K$, identify these subgroups, and estimate the group-specific coefficients $\brho_{1}, \cdots, \brho_{K}$, while  accurately estimating $\bbeta$, as $\bbeta$ governs how effectively the effects of the covariates $\bx_i$ are controlled and, in turn, influences the accuracy of the subgroup analysis.  To this end, we propose a  novel and robust approach  that
simultaneously determines the number of subgroups $K$, uncovers the subgroup structure $G$, and estimates 
both covariate effects $\balpha_i$ and $\bbeta$.

\subsection{Our Method}
\setcounter{equation}{0}

In this subsection, we introduce our robust subgroup analysis method that integrates inverse probability weighting, M-estimation, and pairwise fusion penalization.
Assume that the censoring time $C$ is a continuous random variable with survival function $S(t) = P(C > t)$.
Denote by $\widehat{S}(t)$ the Kaplan-Meier estimator of $S(t)$ based on 
$\left\{(T_i^{*}, \Delta_i)\right\}_{i=1}^n$.
Write $\balpha=(\balpha_1^{\top}, \cdots, \balpha_n^{\top})^{\top}\in\mathbb{R}^{nq}$. Motivated by the M-estimation approach in \cite{huber1973robust}, 
we consider the following penalized objective function
\begin{align}\label{eq: IPW-obj-fun}
	Q_n(\balpha, \bbeta; \lambda)
	=\sum_{i=1}^n\frac{\Delta_i}{\widehat{S}(T_i^{*})}\rho\left(Y_i^{*}-\bz_i^{\top}\balpha_i-\bx_i^{\top}\bbeta\right) 
	+ \sum_{1\leq i< j\leq n}p(\|\balpha_i - \balpha_j\|, \lambda),
\end{align}
where $Y_i^{*}=\min\{\log(T_i), \,\log(C_i)\}$, $\rho(\cdot)$ is a known convex loss function, $p(\cdot, \lambda)$ is a concave penalty function with a tuning parameter $\lambda\geq 0$. 
Among the various loss functions used for robust M-estimation, the Huber loss provides a desirable balance between robustness and efficiency. Thus, 
in this paper, we adopt the Huber loss function 
\begin{align*}
	\rho(x) =
	\left\{\begin{array}{ll}
		x^2 /2 ,    & \quad\mbox{if } |x | \leq \tau ,  \\
		\tau |x | -  \tau^2 /2 ,   & \quad \mbox{if }  |x | > \tau ,
	\end{array}  \right. 	
\end{align*}
where $\tau>0$ is the robustification parameter.
As noted by~\citet{ma2017concave}, the $L_1$ penalty $p(t,\lambda)=\lambda|t|$ may fail to accurately recover subgroup structures because it applies the same thresholding to all pairwise differences $\|\boldsymbol{\alpha}_i-\boldsymbol{\alpha}_j\|$, which can result in biased estimation.
Therefore, we adopt 
the SCAD penalty~\citep{fan2001variable}, 
a concave penalty defined as
	$p(t, \lambda)=\lambda\int_0^t\min\{1, (a-u/\lambda)_{+}/(a-1)\}\,du$ with $a>2$,
	where $x_{+}=\max\{x, 0\}$ and $a$ 
	controls
	the concavity of the penalty function.
	
	
	Denote by $(\widehat{\balpha} (\lambda), \widehat{\bbeta}(\lambda))$  the minimizer of \eqref{eq: IPW-obj-fun} for a given $\lambda>0$, that is,
	\begin{align}\label{eq: optimization}
		(\widehat{\balpha} (\lambda), \widehat{\bbeta}(\lambda))
		= \underset{\footnotesize \balpha,\, \bbeta}{\arg\min}\, Q_n(\balpha, \bbeta; \lambda).
	\end{align}
	Following~\cite{ma2017concave}, we use the modified BIC to select the penalty parameter $\lambda$ by minimizing
	\begin{align}\label{BIC-lambda}
		\mbox{BIC}(\lambda)= \log\left[n^{-1}\sum_{i=1}^n\frac{\Delta_i}{\widehat{S}(T_i^{*})}\rho\left(Y_i^{*}-\bz_i^{\top}\widehat{\balpha}_i-\bx_i^{\top}\widehat{\bbeta}\right)\right]
		+C_n\frac{\log n}{n}\left(\widehat{K}(\lambda)q+p\right),
	\end{align}
	where $C_n=\log(nq+p)$ is a positive number and $\widehat{K}(\lambda)$ is the estimated number of subgroups.
	Let $\widehat{\lambda}$ denote the value of $\lambda$ selected by the modified BIC criterion in~\eqref{BIC-lambda}.
	Write
	$(\widehat{\balpha}, \widehat{\bbeta})=(\widehat{\balpha} (\widehat{\lambda}), \widehat{\bbeta}(\widehat{\lambda}))$.
	Denote by $\widehat{\brho}_1, \cdots, \widehat{\brho}_{\widehat{K}}$ the distinct values of $\widehat{\balpha}_1, \cdots, \widehat{\balpha}_n.  $
	Then the corresponding subgroups are defined as
	$\widehat{G}_{k}=\{i: \widehat{\balpha}_i=\widehat{\brho}_k, 1\leq i\leq n\}$ for $k=1, \cdots, \widehat{K}$.
		\begin{remark}
			Rank-based methods and weighted estimating equations are two widely used
			approaches for robust estimation under censoring; see, for example, \cite{wei1990linear}, \cite{ying1993large}, \cite{jin2007m}, and \cite{huang2007least}. These approaches provide effective tools for handling heavy-tailed errors and outliers in censored regression models. However, they are mainly developed for homogeneous regression settings and focus on estimating 
			the regression parameter. Relatively few studies have extended such robust methods to subgroup analysis. For example, \cite{fu2021efficient} proposed a rank-based approach for robust estimation and clustering under the accelerated failure time model, but the number of clusters is assumed to be known in advance. In contrast, our approach addresses heterogeneity by simultaneously identifying the unknown number of subgroups and estimating subgroup-specific covariate effects under censoring.
		\end{remark}

\section{Algorithm}\label{section3}
\setcounter{equation}{0}


This section presents our robust inverse-probability-weighted subgroup analysis algorithm, termed RISA-ADMM, for solving the optimization problem \eqref{eq: optimization} and obtaining $(\widehat{\boldsymbol{\alpha}}, \widehat{\boldsymbol{\beta}})$. By introducing a new set of parameters $\bet_{ij}=\balpha_i-\balpha_j$, the optimization problem \eqref{eq: optimization} 
can be reformulated as minimizing
\begin{align*}
	L_0(\balpha,\bbeta,\bet)&=\sum_{i=1}^n\frac{\Delta_i}{\widehat{S}(T_i^{*})}\rho\left(Y_i^{*}-\bz_i^{\top}\balpha_i-\bx_i^{\top}\bbeta\right)
	+ \sum_{1\leq i< j\leq n}p(\|\bet_{ij}\|, \lambda),\\
	& \mbox{subject to} \ \balpha_i-\balpha_j-\bet_{ij}=\bzero,
\end{align*}
where $\bet=\{\bet_{ij}^{\top}, i<j\}^{\top}\in\mathbb{R}^{qn(n-1)/2}$. 
Consequently, the corresponding augmented Lagrangian is
\begin{align*}
	L(\balpha,\bbeta,\bet,\bnu)=L_0(\balpha,\bbeta,\bet)+\sum_{i<j}\langle\bnu_{ij},\balpha_i-\balpha_j-\bet_{ij}\rangle  +\frac{\vartheta}{2}\sum_{i<j}\|\balpha_i-\balpha_j-\bet_{ij}\|^2,
\end{align*}
where the dual variables $\bnu=\{\bnu_{ij}^{\top}, i<j\}^{\top}$ are Lagrangian multipliers and $\vartheta$ is a penalty parameter. 
We then compute the estimators of $(\balpha,\bbeta,\bet,\bnu)$ using the following RISA-ADMM algorithm.
%
	%
	For a given $(\balpha^{(m)},\bbeta^{(m)},\bet^{(m)},\bnu^{(m)})$, which is the value of $(\balpha,\bbeta,\bet,\bnu)$ at iteration $m$, our algorithm proceeds as follows:
	\begin{align}
		\mbox{Step 1}: &\quad  (\balpha^{(m+1)},\bbeta^{(m+1)}) = \underset{\footnotesize\balpha,\,\bbeta}{\arg\min}\, L(\balpha,\bbeta,\bmeta^{(m)},\bnu^{(m)}),\label{IPW-step1}\\
		\mbox{Step 2}: & \quad \bet^{(m+1)} = \underset{\footnotesize\bmeta}{\arg\min}\, L(\balpha^{(m+1)},\bbeta^{(m+1)},\bet,\bnu^{(m)}),\label{IPW-step2}\\
		\mbox{Step 3}:  &\quad  \bnu_{ij}^{(m+1)} =\bnu_{ij}^{(m)}+\vartheta(\balpha_i^{(m+1)}-\balpha_j^{(m+1)}-\bet_{ij}^{(m+1)})\label{IPW-step3}.
	\end{align}
	
	The minimization problem in Step 1 is equivalent to minimizing the following function
	\begin{align}\label{IPW-obj-f}
		f(\balpha,\bbeta)=\sum_{i=1}^n\frac{\Delta_i}{\widehat{S}(T_i^{*})}\rho\left(Y_i^{*}-\bz_i^{\top}\balpha_i-\bx_i^{\top}\bbeta\right)
		+\frac{\vartheta}{2}\|\bA\balpha-\bet^{(m)}+\vartheta^{-1}\bnu^{(m)}\|^2,
	\end{align}
	where $\bA=\bD\otimes\bI_q$, 
	$\bD=\{(\bbe_i-\bbe_j),i<j\}^{\top}\in\mathbb{R}^{(n(n-1)/2)\times n}$, $\bbe_i$ is an $n\times1$ unit vector with $1$ in the 
	$i$th entry and $0$ elsewhere, 
	$\bI_q$ is a $q\times q$ identity matrix, and $\otimes$ is the Kronecker product. 
	Because the objective function $f(\balpha,\bbeta)$ in \eqref{IPW-obj-f} involves the loss function 
	$\rho(\cdot)$, the algorithms in~\cite{ma2017concave} and \cite{yan2021subgroup} are not directly applicable. 
	To address this issue, we borrow ideas from~\cite{cai2021robust} and adopt an iterative penalized reweighted least-squares (IPRLS) method to minimize~\eqref{IPW-obj-f}.
	%
	To proceed, 
	define $\bpsi_i = (\bzero^{\top}, \ldots, \bz_i^{\top}, \ldots, \bx_i^{\top})^{\top}$, a sparse $(nq + p)$-dimensional column vector with (n+1) blocks, where the $i$th block equals $\bz_i$, the last block is $\bx_i$, and all other blocks are zero.
	Write $\bPsi=(\bpsi_1,\cdots,\bpsi_n)^{\top}$ and $\bgamma=(\balpha^{\top},\bbeta^{\top})^{\top}$. Setting the derivative of the objective function \eqref{IPW-obj-f} with respect to $\bgamma$ to zero yields the following equation 
	\begin{align*}
		-\sum_{i=1}^{n}\frac{\Delta_i}{\widehat{S}(T_i^*)}\dot{\rho}(Y_i^*-\bpsi_i^{\top}\bgamma)\bpsi_i^{\top}+\vartheta(\widetilde{\bA}\bgamma-\bet^{(m)}+\vartheta^{-1}\bnu^{(m)})^{\top}\widetilde{\bA}=\bzero,
	\end{align*}
	where 
	$\widetilde{\bA}=(\bA,\bzero)$ is a $[nq(n-1)/2]\times (nq+p)$ matrix. For convenience, let $r_i=Y_i^*-\bpsi_i^{\top}\bgamma$, $\pi_i=\dot{\rho}(r_i)/r_i$, $\widetilde{\pi}_i=[\Delta_i/\widehat{S}(T_i^*)]\pi_i$, and $\widetilde{\bPi}=\diag\{\widetilde{\pi}_1,\cdots, \widetilde{\pi}_n\}$. Then the above equation can be rewritten as 
	\begin{align*}
		- \bPsi^{\top}\widetilde{\bPi}(\bY^{*}-\bPsi\bgamma)
		+\vartheta\widetilde{\bA}^{\top}(\widetilde{\bA}\bgamma-\bet^{(m)}+\vartheta^{-1}\bnu^{(m)})=\bzero,
	\end{align*}
	where $\bY^{*}=(Y_1^{*}, \cdots, Y_n^{*})^{\top}$. This suggests that the estimators of $\balpha$ and $\bbeta$ at  iteration $(m+1)$ can be updated by
	\begin{align}\label{IPW-gamma}
		&({\balpha^{(m+1)}}^{\top}, {\bbeta^{(m+1)}}^{\top})^{\top}
		= \bgamma^{(m+1)} \nonumber\\
		=& \left(\vartheta\widetilde{\bA}^{\top}\widetilde{\bA}+\bPsi^{\top}\widetilde{\bPi}^{(m)}\bPsi\right)^{-1}(\vartheta\widetilde{\bA}^{\top}\bet^{(m)}-\widetilde{\bA}^{\top}\bnu^{(m)}+\bPsi^{\top}\widetilde{\bPi}^{(m)}\bY^{*}),
	\end{align}
	where $\widetilde{\bPi}^{(m)}$ is the value of $\widetilde{\bPi}$ evaluated at $\bgamma^{(m)}=({\balpha^{(m)}}^{\top}, {\bbeta^{(m)}}^{\top})^{\top}$. 
	
	
	In Step 2, after discarding the terms independent of $\bet$, we need to minimize
	\begin{align}\label{IPW-eta}
		\frac{\vartheta}{2}\|\bzeta_{ij}^{(m)}-\bet_{ij}\|^2+p(\|\bet_{ij}\|,\lambda),
	\end{align}
	where $\bzeta_{ij}^{(m)}=\balpha_i^{(m+1)}-\balpha_j^{(m+1)}+\vartheta^{-1}\bnu_{ij}^{(m)}$. 
	For the SCAD penalty with $a>1/\vartheta+1$, the solution of \eqref{IPW-eta} is 
	\begin{equation*}
		\bet_{ij}^{(m+1)} =
		\left\{
		\begin{array}{ll}
			\mbox{ST}(\bzeta_{ij}^{(m)}, \lambda /\vartheta) & \rm{if} \ |\bzeta_{ij}^{(m)}\| \leq \lambda+\lambda/\vartheta, \\
			\frac{\mbox{ST}(\bzeta_{ij}^m,a\lambda/((a-1)\vartheta))}{1-1/((a-1)\vartheta)}&\rm{if}\, \lambda+\lambda/\vartheta<\|\bzeta_{ij}^{(m)}\|\leq a\lambda,\\
			\bzeta_{ij}^{(m)} & \rm{if}\, \|\bzeta_{ij}^{(m)}\| > a \lambda, 
		\end{array}
		\right.
	\end{equation*}
	where $\mbox{ST}(\bu,t)=(1-t/\|\bu\|)_{+}\bu$ is the groupwise soft thresholding operator.
	
	
	In Step 3, we update $\bnu_{ij}$ using  \eqref{IPW-step3}. 
	The iterations are terminated
	when the primal residual $\bA\balpha^{(m+1)}-\bet^{(m+1)}$ is sufficiently close to zero, that is, $\|\bA\balpha^{(m+1)}-\bet^{(m+1)}\|\leq \epsilon$ for some small value $\epsilon>0$. Once convergence is reached, subjects $i$ and $j$ with $\widehat{\bet}_{ij}=0$ can be grouped into one subgroup $G_k$. In addition, we can estimate the $k$th subgroup-specific 
	effect using $\widehat{\brho}_k=|\widehat{G}_k|^{-1}\sum_{i\in \widehat{G}_k}\widehat{\balpha}_i$, where $|\widehat{G}_k|$ is the cardinality of $\widehat{G}_k$. The pseudo-code for our RISA-ADMM algorithm is presented in Algorithm \ref{algorithm1}. 
	
	\vspace{4mm}
	
	\begin{algorithm}
		\caption{\enskip RISA-ADMM}\label{algorithm1}
		\begin{algorithmic}[1]
			\State \textbf{Input:} 
			(1) $\bz_i$, $\bx_i$, $T_i^*$, $\Delta_i$ for $i=1,\dots,n$. 
			(2) constants $a$, $\vartheta$, $\tau$, termination accuracy $\epsilon$, and tuning parameter $\lambda$.
			\State Compute matrices $\widetilde{\bA}$, $\bPsi$, and Kaplan-Meier estimators $\widehat{S}(T_i^*)$ for all $i$.
			\State Let the initial matrix $\widetilde{\bPi} = \diag\{\widetilde{\pi}_1, \dots, \widetilde{\pi}_n\}$ with $\widetilde{\pi}_i = \Delta_i / \widehat{S}(T_i^*)$.
			\State \textbf{Initialize:} 
			Obtain initial estimates $\balpha^{(0)}$ and $\bbeta^{(0)}$; 
			set $\bet_{ij}^{(0)} = \balpha_i^{(0)} - \balpha_j^{(0)}$ and $\bnu_{ij}^{(0)} = \bzero$ for all $(i,j)$.
			\State Set iteration index $m \gets 0$.
			\Repeat
			\State Update $\widetilde{\bPi}^{(m)}$ as the value of $\widetilde{\bPi}$ evaluated at $\bgamma^{(m)} = \big((\balpha^{(m)})^{\top}, (\bbeta^{(m)})^{\top}\big)^{\top}$.
			\State Update $\balpha^{(m+1)}$ and $\bbeta^{(m+1)}$ according to \eqref{IPW-gamma}.
			\State Compute $\bzeta_{ij}^{(m)} = \balpha_i^{(m+1)} - \balpha_j^{(m+1)} + \vartheta^{-1}\bnu_{ij}^{(m)}$.
			\State Update $\bmeta_{ij}^{(m+1)}$ by minimizing \eqref{IPW-eta} under the specified penalty.
			\State Update $\bnu_{ij}^{(m+1)}$ according to \eqref{IPW-step3}.
			\State $m \gets m + 1$
			\Until{$\|\bA \balpha^{(m)} - \bet^{(m)}\| \le \epsilon$}.
			\State \textbf{Output:} $(\widehat{\balpha}, \widehat{\bbeta}, \widehat{\bmeta}, \widehat{\bnu}) \gets (\balpha^{(m)}, \bbeta^{(m)}, \bmeta^{(m)}, \bnu^{(m)})$.
		\end{algorithmic}
	\end{algorithm}
	
	\vspace{-4mm}
		\begin{remark}
			It is worth mentioning that the proposed RISA-ADMM algorithm is not restricted to the Huber loss. The iterative updates in equation \eqref{IPW-gamma} are derived via the IPRLS framework, which only requires the first derivative of the loss function to construct the corresponding weights. Therefore, any convex or nonconvex robust loss function with a well-defined score function, such as Tukey’s bisquare loss or other M-estimators, can be incorporated into the proposed framework without modifying the ADMM structure. The Huber loss is adopted in our numerical studies for its simplicity and well-known efficiency-robustness tradeoff, but the proposed methodology itself is more general.
		\end{remark}

	To obtain the initinal value
	$(\balpha^{(0)},  \bbeta^{(0)})$,
	we consider the ridge fusion 
	given by
	\begin{align*}
		L_R(\balpha,\bbeta)=\sum_{i=1}^n\frac{\Delta_i}{\widehat{S}(T_i^{*})}\rho\left(Y_i^{*}-\bz_i^{\top}\balpha_i-\bx_i^{\top}\bbeta\right)+\frac{\lambda^*}{2}\sum_{i<j}\|\balpha_i-\balpha_j\|^2,
	\end{align*}
	where $\lambda^*$ is a small tuning parameter. 
	Using arguments analogous to those in~\citet{cai2021robust},
	we obtain 
	\begin{align}\label{choice}
		(\balpha_R(\lambda^*)^{\top}, \bbeta_R(\lambda^*)^{\top})^{\top}=\bgamma_R(\lambda^*)=(\lambda^*\widetilde{\bA}^{\top}\widetilde{\bA}+\bPsi^{\top}\widetilde{\bPi^*}\bPsi)^{-1}\bPsi^{\top}\widetilde{\bPi^*}\bY^*,
	\end{align}
	where $\widetilde{\pi}^*_i=\Delta_i/\widehat{S}(T_i^*)$, and $\widetilde{\bPi}^*=\diag\{\widetilde{\pi}^*_1,\ldots, \widetilde{\pi}^*_n\}$. 
	We set $\lambda^*=0.001$ in our numerical studies.
	Based on $\balpha_R(\lambda^*)$ obtained from \eqref{choice}, we assign the subjects into $K^*$ preliminary groups according to the ranking of the median values of $\balpha_{R,\,i}(\lambda^*)$. Following \citet{ma2020exploration}, we set $K^*= \lfloor n^{1/2} \rfloor$ to ensure that the number of groups is sufficiently large, where $\lfloor a \rfloor$ denotes the greatest integer not exceeding $a$.
	We then compute the initial estimate $(\balpha^{(0)}, \bbeta^{(0)})$ via least squares regression using the inverse probability weighted M-estimation within the $K^*$ groups and take
	$\bmeta_{ij}^{(0)} = \balpha_i^{(0)} - \balpha_j^{(0)}$ and $\bnu_{ij}^{(0)} = \bzero$. 
	The following results guarantee the convergence of our RISA-ADMM algorithm.  
	
	\begin{prop}\label{proposition1}
		Let $\br^{(m+1)}=\bA\balpha^{(m+1)}-\bet^{(m+1)}$ and $\bs^{(m+1)}=\vartheta\bA^{\top}(\bet^{(m+1)}-\bet^{(m)})$ be the primal residual and the dual residual in the algorithm described above, respectively. It holds that $\lim\limits_{m\to \infty}\|\br^{(m+1)}\|=0$ and $\lim\limits_{m\to \infty}\|\bs^{(m+1)}\|=0$ 
		for the SCAD penalty.
	\end{prop}

	\section{Theoretical results}\label{section4}
	\setcounter{equation}{0}

	To facilitate our theoretical analysis, we first present some notation and conditions.
	Denote by $\widetilde{\bW}=(w_{ik})$ the $n\times K$ matrix with $w_{ik}=1$ for $i\in G_k$ and $w_{ik}=0$ otherwise. Write $\bW=\widetilde{\bW}\otimes \bI_q$. Let $\Omega=\{\balpha\in \mathbb{R}^{nq}:\balpha_i=\balpha_j$, for any $i,j\in G_k, 1\le k\le K\}$. Define $\bphi=(\bbeta^{\top},\brho^{\top})^{\top}$ and $\brho=(\brho_1^{\top},\dots,\brho_K^{\top})^{\top}$, where $\brho_k$ is the $k$th subgroup-specific vector. 
	Let $\bphi_0=(\bbeta_0^{\top},\brho_0^{\top})^{\top}$ be the true parameter vector of $\bphi$.
	Denote by $|G_k|$ the number of elements in $G_k$. Define  $G_{\min}=\min_{1\le k\le K}|G_k|$ and $G_{\max}=\max_{1\le k\le K}|G_k|$. Let $\bX=(\bx_1,\dots,\bx_n)^{\top}$ and $\bZ=\diag(\bz_1^{\top},\dots,\bz_n^{\top})$. Let $\bU=(\bX,\bZ\bW)$ and $\bU_i$ be the $i$th row vector of $\bU$; that is, $\bU_i=(\bx_i^{\top},\bz_i^{\top}w_{i1},\dots,\bz_i^{\top}w_{iK})^{\top}$. Then we define
	\begin{align*}
		&\bV_n=\sum_{i=1}^n\frac{\Delta_i}{S(T_i^{*})}\ddot{\rho}(Y_i^{*}-\bU_i^{\top}\bphi_0)\bU_i\bU_i^{\top} \\
		&\bSigma_n=\sum_{i=1}^{n}\frac{\Delta_i}{S^2(T_i^{*})}\{\dot{\rho}(Y_i^{*}-\bU_i^{\top}\bphi_0)\}^2\bU_i\bU_i^{\top}.
	\end{align*}
	Let $\bmathcalV_n=\mathbb{E}(\bV_n^{-1}\bSigma_n\bV_n^{-1})$. 
	Corresponding to the decomposition of $\bU$, we decompose $\bmathcalV_n$ as
	\begin{align*}
		\bmathcalV_n=\left(
		{ \begin{array}{cc}
				\bmathcalV_{n11}&\bmathcalV_{n12}\\
				\bmathcalV_{n21}&\bmathcalV_{n22}\\
		\end{array} }
		\right),
	\end{align*}
	where $\bmathcalV_{n11}$ is a $p\times p$ matrix. 

	
	\begin{condition}\label{c1}
		The function $p_{\lambda}(\theta,\lambda)$ is symmetric with respect to $\theta$ and is non-decreasing and concave in $\theta$ on $[0,\infty)$. Define $\varrho(\theta)=\lambda^{-1}p(\theta,\lambda)$. Then there exists a constant $c>0$ such that $\varrho(\theta)$ is constant for all $\theta \ge c\lambda$. Furthermore, $\varrho(0)=0$, the derivative $\dot{\varrho}(\theta)$ exists and is continuous except at  a finite number of values of $\theta$, and $\dot{\varrho}(0+)=1$.
	\end{condition}

	\begin{condition}\label{c2}
		There exists a constant $\nu>0$ such that $\mathbb{P}(C=\nu)>0$ and $\mathbb{P}(C>\nu)=0$.   
	\end{condition} 
	
	\begin{condition}\label{c3}
		The loss function $\rho(\cdot)$ is assumed to have bounded derivatives of first and second order. In particular, its first derivative is continuous and bounded.
		In addition, the errors $\varepsilon_i$ are independent and identically distributed and satisfy $ \mathbb{E}[\dot{\rho}(\varepsilon_i)] = 0$.
	\end{condition}
	
	\begin{condition}\label{c4}
		$\sup_i\|\bx_i\|\le c_1\sqrt{p}$, $\sup_i\|\bz_i\|\le c_2 \sqrt{q}$, $\Lambda_{\min}(\bU^{\top}\bU)\ge c_3 G_{\min}$, and $\Lambda_{\max}(\bU^{\top}\bU)\le c_4 n$ for some positive constants $c_1, c_2, c_3$ and $c_4$. 
	\end{condition}

	
	Condition~\ref{c1} is standard in the subgroup analysis literature~\citep{ma2017concave, hu2021subgroup,yan2021subgroup, wu2026simultaneous} and
	is satisfied by many concave penalty functions, including the SCAD penalty.  
		Condition 2 is a technical assumption that ensures the censoring survival function remains bounded away from zero, thereby preventing extreme inverse probability weights and simplifying asymptotic arguments. This condition is commonly adopted in the censored data literature (e.g.,~\cite{peng2009competing, song2014censored, fang2020joint}) and is often satisfied in clinical studies with administrative censoring. Moreover, as noted by~\cite{peng2009competing}, one may consider a truncated censoring time $C^* = \min(C, L)$, which guarantees that the condition is always satisﬁed. In practice, $L$ should be chosen below the upper bound of the support of $C$ but sufficiently large so that any resulting loss of information is negligible.
	Condition~\ref{c3} is commonly assumed in the literature on M-estimation~\citep{huber1973robust}. 
	In particular, higher-order derivatives are required for Taylor expansions and the assumption 
	$\mathbb{E}[\dot{\rho}(\varepsilon_i)] = 0$
	ensures that 
	$\sum_{i=1}^n\left [\Delta_i/S(T_i)\right]\rho(Y_i - \mathbf{U}_i^\top \boldsymbol{\phi})$
	is minimized at the true parameter value $\boldsymbol{\phi}_0$. Condition~\ref{c4} imposes a mild restriction on the $L_2$-norms of $\mathbf{x}_i$ and $\mathbf{z}_i$, which relaxes the fixed upper-bound assumption on these norms used in \cite{hu2021subgroup} and \cite{yan2021subgroup}.  The assumptions on the matrix $\mathbf{U}$ are similar to those in \cite{ma2020exploration} and \cite{yan2021subgroup}.
	
	Next, we introduce the oracle estimators, 
	assuming that
	the true subgroup structure is known. 
	Although the oracle estimators are not available in practice, they are useful for deriving the theoretical properties of our proposed method.
	Given the true subgroup structure $\bU$, 
	the oracle estimators of $\bbeta$ and $\balpha$
	are defined as
	\begin{align*}
		( \widehat{\balpha}^{or}, \widehat{\bbeta}^{or})
		=\underset{\footnotesize \balpha\in\Omega,\,\bbeta\in \mathbb{R}^{p}}{\arg\min}\,\sum_{i=1}^n\frac{\Delta_i}{\widehat{S}(T_i^{*})}\rho(Y_i^{*}-\bz_i^{\top}\balpha_i-\bx_i^{\top}\bbeta).
	\end{align*}
	Let $\widehat{\brho}^{or}$ be the distinct values of $\widehat{\balpha}^{or}$. The oracle estimator of $\bphi$ is then defined as $\widehat{\bphi}^{or}=\big((\widehat{\bbeta}^{or})^{\top}, (\widehat{\brho}^{or})^{\top}\big)^{\top}$. 
	Equivalently,  the oracle estimator of $\bphi$ is also given by 
	\begin{align*}
		\widehat{\bphi}^{or}=\underset{\footnotesize \bphi\in \mathbb{R}^{p+Kq}}{\arg\min}\,\sum_{i=1}^n\Delta_i/\widehat{S}(T_i^{*})\rho(Y_i^{*}-\bU_i^{\top}\bphi).
	\end{align*}

	
		\begin{theo}\label{theorem1}
			Assume that Conditions~\ref{c1}-\ref{c4} and
			\begin{align*}
						\mathbb{P}\left(\lim_{n\rightarrow \infty} \left\{\inf_{\footnotesize\| \bphi-\bphi_0\|\ge n^{-\gamma}}\|\Phi_n^S(\bphi)\| \cdot\left(n^{1/2+r}\sqrt{Kq+p}\right)^{-1}\right\}=\infty\right)=1  
			\end{align*}
			hold,  where $\Phi_n^S(\bphi)=\sum_{i=1}^n [\Delta_i/S(T_i^{*})]\dot{\rho}(Y_i^{*}-\bU_i^{\top}\bphi)\bU_i$, $\gamma>0$ and $0<r<1/2$.  If $G_{\min}\gg\max\{n^{1-\gamma},n^{1/2+r}\sqrt{Kq+p}\}$, 
			then we have
			\begin{align*}
				&\|\big((\widehat{\bbeta}^{or}-\bbeta_0)^{\top},(\widehat{\brho}^{or}-\brho_0)^{\top}\big)^{\top}\|=o(\xi_{n}),\\
				&	\|\widehat{\boldsymbol{\alpha}}^{or}-\boldsymbol{\alpha}_0\| = o\big(\sqrt{G_{\max}} \, \xi_n\big), \,\,\mbox{and}\,\, 
				\sup_i \|\widehat{\boldsymbol{\alpha}}^{or}_i - \boldsymbol{\alpha}_{0i}\| = o(\xi_n)
			\end{align*}
			almost surely,
			where 
			\begin{align}\label{eq: xin-def}
				\xi_n \;=\; \max\left\{\frac{n^{1-\gamma}}{G_{\min}},\;
				\frac{n^{1/2+r}\sqrt{Kq+p}}{G_{\min}}\right\}.  
			\end{align}
			Moreover, we have 
			\begin{align*} 
				\ba_n^{\top}\bmathcalV_n^{-1/2}
				\big((\widehat{\bbeta}^{or}-\bbeta_0)^{\top},(\widehat{\brho}^{or}-\brho_0)^{\top}\big)^{\top}
				\xrightarrow{d} N(0,1)
			\end{align*}
			for any unit column vector $\ba_n\in\mathbb{R}^{p+Kq}$ as $n\to\infty$.
		\end{theo}

		Theorem \ref{theorem1} establishes an estimation error bound for the oracle estimator $(\widehat{\bbeta}^{or}, \widehat{\brho}^{or})$ and  further shows that this oracle estimator is asymptotically normally distributed.   
		Under the condition $G_{\min}\gg\max\{n^{1-\gamma},n^{1/2+r}\sqrt{Kq+p}\}$, it follows from \eqref{eq: xin-def} that the estimation error bound $\xi_n$ converges to zero as $n\to\infty$.
		Consequently, the oracle estimator can achieve consistency with probability approaching one.
		

		\begin{remark}\label{remark3}
			The condition on $\Phi_n^S(\bphi)$ in Theorem~\ref{theorem1}, that is, 
			\begin{align*}
				\mathbb{P}\left(\lim_{n\rightarrow \infty} \left\{\inf_{\footnotesize\| \bphi-\bphi_0\|\ge n^{-\gamma}}\|\Phi_n^S(\bphi)\| \cdot\left(n^{1/2+r}\sqrt{Kq+p}\right)^{-1}\right\}=\infty\right)=1,
			\end{align*}
			ensures that $\Phi_n^S(\boldsymbol{\phi})$ diverges from zero 
			outside a shrinking neighborhood of $\boldsymbol{\phi}_0$.
			Under this condition, any solution to $\Phi_n^S(\boldsymbol{\phi})=0$ is expected to lie in a shrinking neighborhood of $\boldsymbol{\phi}_0$, which facilitates the consistency results established in Theorem~\ref{theorem1}. 
			A similar condition was also adopted in \cite{yan2021subgroup}, who considered the same problem
			but used imputation to handle censored outcomes.
			This condition is satisfied under a convexity assumption on the estimating function $\Phi_n^S(\boldsymbol{\phi})$.
			Specifically, assume 
			that there exists a constant $c>0$ such that
			$\bigl(\Phi_n^S(\boldsymbol{\phi})-\Phi_n^S(\boldsymbol{\phi}_0)\bigr)^\top
			(\boldsymbol{\phi}-\boldsymbol{\phi}_0)
			\ge c n \|\boldsymbol{\phi}-\boldsymbol{\phi}_0\|^2$
			holds for all $\boldsymbol{\phi}$ satisfying $\|\boldsymbol{\phi}-\boldsymbol{\phi}_0\|\ge n^{-\gamma}$. 
			Then, using the Cauchy–Schwarz inequality and the reverse triangle inequality, it is straightforward to show that
			the condition $ \mathbb{P}\left(\lim\limits_{n\to\infty}
			\left\{\inf_{\|\boldsymbol{\phi}-\boldsymbol{\phi}_0\|\ge n^{-\gamma}}
			\|\Phi_n^S(\boldsymbol{\phi})\|
			\cdot \left(n^{1/2+r}\sqrt{Kq+p}\right)^{-1}\right\}
			= \infty \right) = 1$ holds whenever $n^{\frac12-\gamma-r}\gg \sqrt{Kq+p}$.
		\end{remark}

		\begin{remark}\label{remark4}
			Since $G_{\min} \le n/K$ by definition, we may write $G_{\min} = \delta n/K$ for some constant $0<\delta \leq 1$. When the number of subgroups $K$ is fixed, the condition $G_{\min}\gg\max\{n^{1-\gamma},\,n^{1/2+r}\sqrt{Kq+p}\}$ in Theorem~1 implies
			$p=o(n^{1-2r})$ and $q=o(n^{1-2r})$.
			Therefore, both $p$
			and $q$ are allowed to diverge with $n$ when 
			$K$ is fixed.
		\end{remark}


	%
	
	For $K\ge 2$, we define the minimum signal difference between the heterogeneous effects of any two subgroups as
	\begin{align*}
		b_n=\min_{i\in G_k,\, j\in G_{k'},\, k\ne k^{'}}\|\balpha_{0i}-\balpha_{0j}\|=\min_{k\ne k^{'}}\|\brho_{0k}-\brho_{0k^{'}}\|.
	\end{align*}
	Then we have the following result.
	
	%
	
	\begin{theo}\label{theorem2}
		Suppose Conditions~\ref{c1}-\ref{c4} hold. If $b_n>c\lambda$ and $\lambda\gg \xi_n$, where $c>0$ is the constant specified  in Condition~\ref{c1} and $\xi_n$ is given in~\eqref{eq: xin-def}, then there exists a local minimizer $(\widehat{\balpha}(\lambda),\widehat{\bbeta}(\lambda))$ of the objective function $Q_n(\balpha,\bbeta,\lambda)$ such that 
		\begin{align}\label{eq: consistency}
			\mathbb{P}\left\{\big(\widehat{\balpha}(\lambda),\widehat{\bbeta}(\lambda)\big)=\big(\widehat{\balpha}^{or},\widehat{\bbeta}^{or}\big)\right\}\rightarrow 1\quad\mbox{as}\,\,\,n\to\infty.
		\end{align}
	\end{theo}
	
	Theorem~\ref{theorem2} shows that the oracle estimator $(\widehat{\balpha}^{or},\widehat{\bbeta}^{or})$ is indeed a local minimizer of the objective function~\eqref{eq: IPW-obj-fun} with probability approaching one.
	Let $\widehat{\brho}(\lambda)$ be the distinct values of $\widehat{\balpha}(\lambda)$. 
	Recall that $\widehat{\brho}^{or}$ consists of the distinct values of $\widehat{\balpha}^{or}$. 
	Then it follows from~\eqref{eq: consistency} that 
	$\mathbb{P}\big(\widehat{\brho}(\lambda)=\widehat{\brho}^{or}\big)\rightarrow 1$ and 
	$\mathbb{P}\big(\widehat{\bbeta}(\lambda)=\widehat{\bbeta}^{or}\big)\rightarrow 1$ as $n\to\infty$.
	Combined with Theorem~\ref{theorem1}, which establishes that the
	oracle estimator asymptotically converges to the true parameter, 
	we conclude that our estimator $(\widehat{\brho}(\lambda),\widehat{\bbeta}(\lambda))$ also converges to the true parameter with probability approaching one.
	This result, together with the asymptotic normality established in Theorem~\ref{theorem1},
	directly implies the asymptotic distribution of our estimator $(\widehat{\brho}(\lambda),\widehat{\bbeta}(\lambda))$, as presented in the following corollary.
	
	\begin{coll}
		Under the conditions in Theorem \ref{theorem2}, we have that $ \ba_n^{\top}\bmathcalV_n^{-1/2}(\widehat{\bphi}-\bphi_0)\xrightarrow{d} N(0,1)$ for any unit column vector $\ba_n\in\mathbb{R}^{p+Kq}$ as $n\to\infty$.
		Consequently, we have $\ba_{n1}^{\top}\bmathcalV_{n11}^{-1/2}(\widehat{\bbeta}-\bbeta_0)\xrightarrow{d} N(0,1)$ and $\ba_{n2}^{\top}\bmathcalV_{n22}^{-1/2}(\widehat{\brho}-\brho_0)\xrightarrow{d} N(0,1)$, where
		$\ba_{n1}\in \mathbb{R}^{p}$ and $\ba_{n2}\in \mathbb{R}^{ Kq}$ are unit column vectors.
	\end{coll}

\section{Simulation  studies}\label{section5}
\setcounter{equation}{0}

%

	This section evaluates the finite-sample performance of the proposed method through five simulation studies. Examples 1 and 2 
	involve single and multiple heterogeneous effects, respectively. Example 3 evaluates the performance of our method when $K=1$, corresponding to the homogeneous AFT model. Example 4 investigates the sensitivity of the method to the choice of $\tau$, whereas Example 5 considers a more challenging scenario with a larger number of groups and higher-dimensional covariates $p$ and $q$.
	Due to page limitations, we report only the results for Examples 1 and 2 in the main text. The results for Examples 3–5 are presented in Tables~\ref{tablec1}–\ref{example5-table2} and Figures~\ref{fig: RMSE-Example4}–\ref{fig: RMSE-Example5}  in the Supplementary Material.

\subsection{Example 1 (Single heterogeneous effect)}

Following setups similar to those in~\cite{yan2021subgroup}, 
we generate the failure times  $T_i$ from a heterogeneous AFT model
\begin{align*}
	\log (T_i)=  z_i\alpha_i + \bx_i^{\top}\bbeta + \varepsilon_i,\quad \,\, i=1, 2, \cdots, n,
\end{align*}
where $z_i$ is generated from the standard normal distribution, $\bx_i=(x_{i1},x_{i2})^{\top}$ is generated from a bivariate standard  normal distribution,  $\bbeta=(\beta_1,\beta_2)^{\top}=(-1,1)^{\top}$, and each subgroup-specific coefficient
$\alpha_i$ is drawn from three different values $\rho_1=4$, $\rho_2=-4$, and $\rho_3=0$ with equal probability. Equivalently, there are three subgroups $G_1$, 
$G_2$, and $G_3$ and $\alpha_i=\rho_k$ for $i\in G_k$, $k=1, 2, 3$.
We consider three different cases for generating $\varepsilon_i$:  Case 1 (normal), $\varepsilon_i\stackrel{i . i . d .}{\sim} N(1,0.2^2)$; Case 2 (heteroscedastic normal), $\varepsilon_i=\Phi(x_{i1})\widetilde{\varepsilon}_{i}$ where  $\Phi(\cdot)$ is the cumulative distribution function of $N(0,1)$ and $\widetilde{\varepsilon}_{i}\stackrel{i . i . d .}{\sim}N(0,1)$; Case 3 (t-distribution): $\varepsilon_i\stackrel{i . i . d .}{\sim} 0.5t_{(3)}$.
In addition, we generate the censoring time $C_i$ from $\min(c,U(0,c+2))$, where $U(a, b)$ denotes the uniform distribution 
on the interval $(a, b)$ and $c$ controls the censoring rate. 
We set $c$ to yield a censoring rate of approximately $20\%$ and $40\%$, respectively, and
consider sample sizes of $n=100$ and $n=200$. 
Following~\cite{huber1981robust} and~\cite{zhao2020statistical}, we set 
$\tau=1.345$ for the Huber loss.  
As in \cite{ma2017concave}, we set $a = 3$ for the SCAD penalty and $\vartheta = 1$ in our algorithm. The optimal value of 
$\lambda$ is selected by minimizing the modified BIC criterion  in \eqref{BIC-lambda}. 
%
We compare our method (RISA-ADMM) with the imputation-based approach (BJ-ADMM) proposed by~\cite{yan2021subgroup}.

\begin{table*}[!t]
	\centering 
	\caption{The mean, median, and standard deviation (SD) of $\widehat{K}$, as well as the proportion of correctly identifying the number of subgroups (i.e., $P(\widehat{K} = K)$) for BJ-ADMM and RISA-ADMM over 100 replications in Example 1.}
	\label{tab: Khat-Example1}
	\resizebox{\textwidth}{!}{
		\begin{tabular}{ll*{12}{c}} 
			\toprule
			& & \multicolumn{6}{c}{$n=100$} & \multicolumn{6}{c}{$n=200$} \\
			\cmidrule(lr){3-8} \cmidrule(lr){9-14}
			& & \multicolumn{2}{c}{Case 1} & \multicolumn{2}{c}{Case 2} & \multicolumn{2}{c}{Case 3} & \multicolumn{2}{c}{Case 1} & \multicolumn{2}{c}{Case 2} & \multicolumn{2}{c}{Case 3} \\
			\cmidrule(lr){3-4} \cmidrule(lr){5-6} \cmidrule(lr){7-8} \cmidrule(lr){9-10} \cmidrule(lr){11-12} \cmidrule(lr){13-14}
			Method & Estimation & 20\% & 40\% & 20\% & 40\% & 20\% & 40\% & 20\% & 40\% & 20\% & 40\% & 20\% & 40\% \\
			\midrule
			\multirow{4}{*}{BJ-ADMM} 
			& Mean               & 3.07 & 3.12 & 3.36 & 3.44 & 3.48 & 3.55 & 3.01 & 3.02 & 3.23 & 3.25 & 3.44 & 3.49 \\
			& Median             & 3    & 3    & 3    & 3    & 3    & 3    & 3    & 3    & 3    & 3    & 3    & 3    \\
			& SD                 & 0.26 & 0.36 & 0.66 & 0.74 & 0.62 & 0.67 & 0.10 & 0.14 & 0.46 & 0.48 & 0.58 & 0.63 \\
			& $P(\widehat{K}=K)$ & 0.93 & 0.89 & 0.73 & 0.71 & 0.60 & 0.54 & 0.99 & 0.98 & 0.78 & 0.75 & 0.64 & 0.58 \\
			\addlinespace
			\multirow{4}{*}{RISA-ADMM}
			& Mean               & 3.03 & 3.04 & 3.24 & 3.28 & 3.27 & 3.36 & 3.01 & 3.02 & 3.18 & 3.22 & 3.21 & 3.30 \\
			& Median             & 3    & 3    & 3    & 3    & 3    & 3    & 3    & 3    & 3    & 3    & 3    & 3    \\
			& SD                 & 0.17 & 0.20 & 0.50 & 0.51 & 0.53 & 0.56 & 0.10 & 0.14 & 0.42 & 0.44 & 0.47 & 0.48 \\
			& $P(\widehat{K}=K)$ & 0.97 & 0.96 & 0.79 & 0.75 & 0.78 & 0.72 & 0.99 & 0.98 & 0.85 & 0.82 & 0.81 & 0.78 \\
			\bottomrule
		\end{tabular}
	}
\end{table*}


\begin{figure*}[h!]
	\centering
	\includegraphics[width=\linewidth]{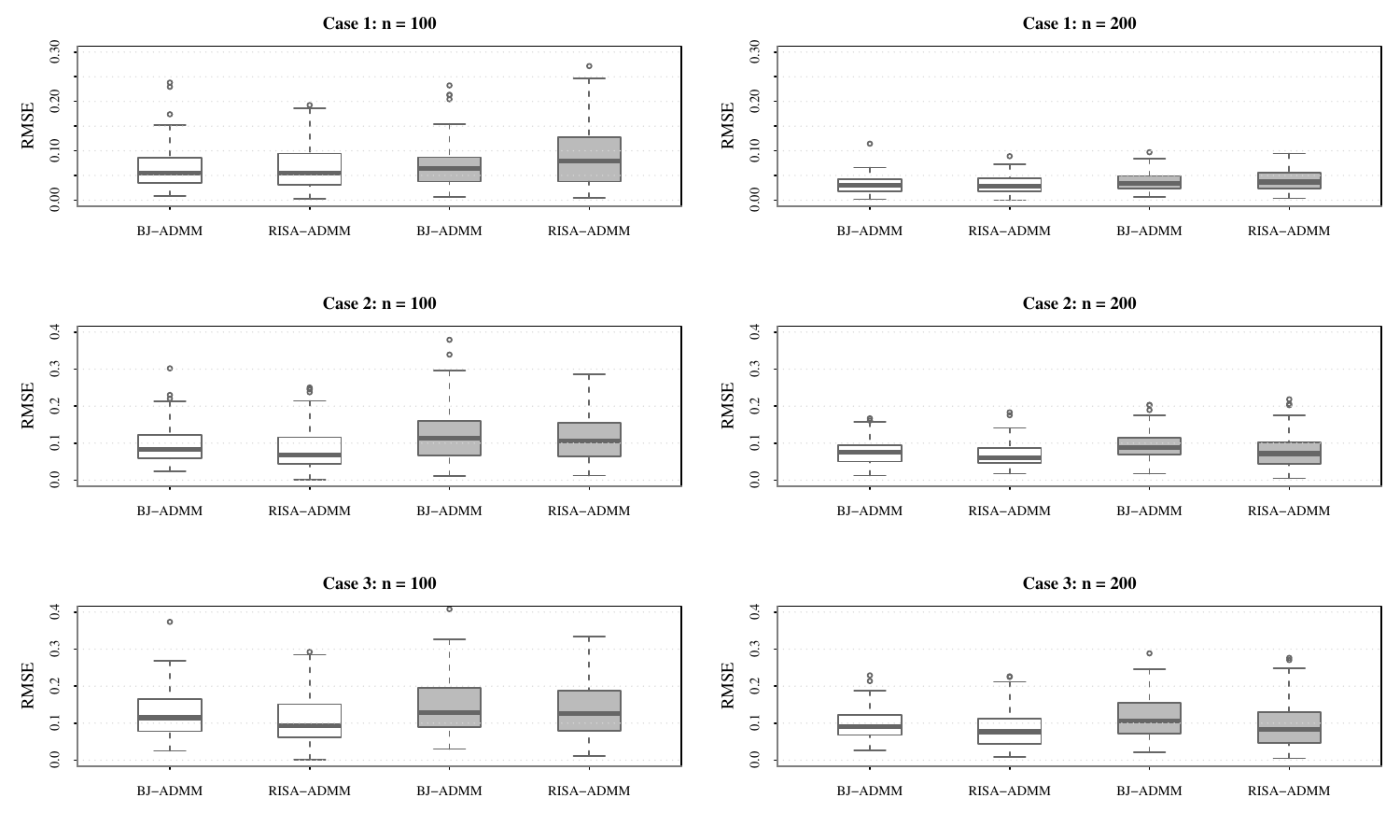}
	
	\caption{Boxplots of the RMSE of $\widehat{\bbeta}$ 
		for BJ-ADMM and RISA-ADMM over 100 replications in Example 1. White: $20\%$ censoring; Grey: $40\%$ censoring.}
	\label{fig: RMSE-Example1}
\end{figure*}

We summarize the results for Example 1 in Table~\ref{tab: Khat-Example1},  Figure~\ref{fig: RMSE-Example1}, and Table~\ref{tab: rho-hat-Example1}.  
These results show that RISA-ADMM performs similarly to BJ-ADMM when the error distribution is normal (Case 1), but consistently outperforms BJ-ADMM by exhibiting greater robustness, lower variability, and more accurate subgroup identification and parameter estimation under heteroscedastic and heavy-tailed error distributions (Cases 2 and 3). We discuss our findings in order.
Table~\ref{tab: Khat-Example1} reports the mean, median, and standard deviation (SD) of the estimated number of subgroups $\widehat{K}$, as well as the proportion of correctly identifying the 
number of subgroups  (i.e., $P(\widehat{K} = K)$) 
for BJ-ADMM and RISA-ADMM under different settings in Example~1 based on 100 replications.  
For both methods, we observe that the median value of $\widehat{K}$ is 3 for all cases, which is the true number of groups in this example.
As expected, increasing the sample size 
improves the performance of both methods, 
as reflected by estimated numbers of subgroups closer to the true value of $K$, reduced standard deviations, and higher probabilities of correct subgroup identification. In contrast, increasing the censoring rate degrades the performance of both methods. It can be seen that our method RISA-ADMM consistently outperforms BJ-ADMM in identifying subgroups across most settings. Specifically, for Case~1, where the error distribution is normal, BJ-ADMM and RISA-ADMM exhibit similar performance across different sample sizes and censoring rates. However, for Cases~2 and~3, which involve heteroscedastic errors or heavier-tailed error distributions, our method RISA-ADMM yields estimated subgroup numbers that are closer to the true number of subgroups, with smaller standard deviations across different sample sizes and censoring rates, indicating greater stability. In addition, in terms of subgroup recovery accuracy, RISA-ADMM achieves higher probabilities of correct subgroup identification than BJ-ADMM in Cases~2 and~3.

\begin{table*}[h!]
	\centering
	\caption{The mean, median, and standard deviation of the estimates $\widehat{\rho}_1$, $\widehat{\rho}_2$, and $\widehat{\rho}_3$ obtained using the BJ-ADMM and RISA-ADMM methods, along with their oracle counterparts, in Example 1.}
	\label{tab: rho-hat-Example1}
	\resizebox{\textwidth}{!}{%
		\begin{tabular}{lll*{12}{c}}
			\toprule
			& & & \multicolumn{6}{c}{$n=100$} & \multicolumn{6}{c}{$n=200$} \\ 
			\cmidrule(lr){4-9} \cmidrule(lr){10-15}
			& & & \multicolumn{3}{c}{Censoring = 20\%} & \multicolumn{3}{c}{Censoring = 40\%} & \multicolumn{3}{c}{Censoring = 20\%} & \multicolumn{3}{c}{Censoring = 40\%} \\ 
			\cmidrule(lr){4-6} \cmidrule(lr){7-9} \cmidrule(lr){10-12} \cmidrule(lr){13-15}
			Case & Method & Estimation & Mean & Median & SD & Mean & Median & SD & Mean & Median & SD & Mean & Median & SD \\ 
			\midrule
			
			\multirow{12.5}{*}{Case 1}
			& \multirow{6}{*}{BJ-ADMM} 
			& $\widehat{\rho}_1$      & 3.992 & 3.998 & 0.066 & 3.986 & 4.000 & 0.068 & 3.997 & 3.996 & 0.036 & 3.993 & 3.993 & 0.037 \\
			& & $\widehat{\rho}_1^{or}$ & 4.001 & 4.003 & 0.047 & 4.006 & 3.997 & 0.051 & 3.997 & 3.985 & 0.033 & 3.994 & 3.985 & 0.034 \\
			& & $\widehat{\rho}_2$      & -3.980 & -3.988 & 0.048 & -3.978 & -3.996 & 0.064 & -3.989 & -3.989 & 0.036 & -3.983 & -3.983 & 0.041 \\
			& & $\widehat{\rho}_2^{or}$ & -4.012 & -4.013 & 0.040 & -4.014 & -4.012 & 0.055 & -3.990 & -3.995 & 0.031 & -4.011 & -3.978 & 0.034 \\
			& & $\widehat{\rho}_3$      & -0.004 & 0.000 & 0.092 & 0.011 & 0.009 & 0.103 & -0.002 & -0.003 & 0.047 & 0.005 & -0.002 & 0.057 \\
			& & $\widehat{\rho}_3^{or}$ & 0.002 & -0.013 & 0.039 & -0.010 & -0.020 & 0.069 & 0.001 & -0.003 & 0.025 & 0.002 & 0.000 & 0.036 \\ 
			\addlinespace[3pt]
			& \multirow{6}{*}{RISA-ADMM} 
			& $\widehat{\rho}_1$      & 3.987 & 3.985 & 0.051 & 3.969 & 3.977 & 0.055 & 3.989 & 3.983 & 0.047 & 3.984 & 3.994 & 0.062 \\
			& & $\widehat{\rho}_1^{or}$ & 4.000 & 4.001 & 0.051 & 3.994 & 4.000 & 0.060 & 4.003 & 4.006 & 0.034 & 4.005 & 4.005 & 0.033 \\
			& & $\widehat{\rho}_2$      & -3.985 & -3.981 & 0.045 & -3.972 & -3.971 & 0.056 & -3.987 & -3.981 & 0.040 & -3.978 & -3.987 & 0.042 \\
			& & $\widehat{\rho}_2^{or}$ & -4.002 & -4.002 & 0.042 & -4.001 & -4.000 & 0.055 & -3.995 & -3.996 & 0.033 & -3.998 & -3.994 & 0.035 \\
			& & $\widehat{\rho}_3$      & 0.004 & 0.005 & 0.044 & 0.009 & 0.011 & 0.051 & 0.001 & 0.003 & 0.033 & 0.005 & 0.003 & 0.041 \\
			& & $\widehat{\rho}_3^{or}$ & 0.006 & 0.005 & 0.036 & 0.003 & 0.004 & 0.043 & 0.001 & 0.004 & 0.028 & -0.003 & -0.002 & 0.040 \\
			\midrule
			
			\multirow{12.5}{*}{Case 2}           
			& \multirow{6}{*}{BJ-ADMM} 
			& $\widehat{\rho}_1$      & 3.932 & 3.817 & 0.264 & 3.914 & 3.765 & 0.293 & 3.985 & 3.978 & 0.092 & 3.977 & 3.993 & 0.117 \\
			& & $\widehat{\rho}_1^{or}$ & 3.995 & 4.005 & 0.163 & 3.973 & 4.019 & 0.189 & 4.004 & 4.005 & 0.082 & 3.983 & 3.991 & 0.104 \\
			& & $\widehat{\rho}_2$      & -3.976 & -3.946 & 0.272 & -3.969 & -3.970 & 0.299 & -3.979 & -3.976 & 0.108 & -3.968 & -3.973 & 0.119 \\
			& & $\widehat{\rho}_2^{or}$ & -3.986 & -3.996 & 0.141 & -3.979 & -3.998 & 0.158 & -3.995 & -3.992 & 0.106 & -3.982 & -3.984 & 0.115 \\
			& & $\widehat{\rho}_3$      & 0.029 & 0.037 & 0.106 & 0.037 & 0.034 & 0.144 & -0.005 & -0.010 & 0.107 & -0.022 & -0.031 & 0.125 \\
			& & $\widehat{\rho}_3^{or}$ & 0.015 & 0.018 & 0.097 & -0.015 & -0.002 & 0.126 & 0.000 & -0.004 & 0.079 & -0.002 & 0.000 & 0.109 \\ 
			\addlinespace[3pt]
			& \multirow{6}{*}{RISA-ADMM} 
			& $\widehat{\rho}_1$      & 3.980 & 3.973 & 0.166 & 3.977 & 3.979 & 0.173 & 3.995 & 3.982 & 0.096 & 3.987 & 3.967 & 0.100 \\
			& & $\widehat{\rho}_1^{or}$ & 3.982 & 3.980 & 0.159 & 3.980 & 3.973 & 0.164 & 4.012 & 3.999 & 0.105 & 4.015 & 3.997 & 0.111 \\
			& & $\widehat{\rho}_2$      & -3.990 & -3.988 & 0.174 & -3.983 & -3.975 & 0.183 & -3.992 & -3.986 & 0.100 & -3.989 & -3.986 & 0.106 \\
			& & $\widehat{\rho}_2^{or}$ & -3.996 & -3.972 & 0.160 & -3.996 & -4.000 & 0.159 & -4.003 & -3.987 & 0.100 & -4.009 & -3.972 & 0.110 \\
			& & $\widehat{\rho}_3$      & 0.007 & 0.010 & 0.121 & 0.019 & 0.022 & 0.194 & 0.004 & 0.002 & 0.087 & -0.014 & -0.013 & 0.122 \\
			& & $\widehat{\rho}_3^{or}$ & 0.007 & 0.020 & 0.099 & 0.009 & 0.013 & 0.154 & -0.001 & -0.003 & 0.076 & -0.001 & -0.013 & 0.105 \\ 
			\midrule
			
			\multirow{12.5}{*}{Case 3}              
			& \multirow{6}{*}{BJ-ADMM} 
			& $\widehat{\rho}_1$      & 3.943 & 3.976 & 0.276 & 3.931 & 3.921 & 0.335 & 3.957 & 3.781 & 0.211 & 3.938 & 3.741 & 0.269 \\
			& & $\widehat{\rho}_1^{or}$ & 3.993 & 3.993 & 0.193 & 4.048 & 4.013 & 0.197 & 4.006 & 4.022 & 0.118 & 4.003 & 4.012 & 0.124 \\
			& & $\widehat{\rho}_2$      & -3.981 & -3.797 & 0.251 & -3.967 & -3.783 & 0.311 & -3.989 & -3.762 & 0.220 & -3.973 & -3.684 & 0.285 \\
			& & $\widehat{\rho}_2^{or}$ & -4.013 & -4.007 & 0.158 & -4.016 & -4.019 & 0.199 & -4.010 & -4.022 & 0.120 & -4.022 & -4.008 & 0.129 \\
			& & $\widehat{\rho}_3$      & -0.002 & -0.005 & 0.197 & -0.028 & -0.025 & 0.199 & -0.008 & -0.009 & 0.114 & -0.011 & -0.010 & 0.165 \\
			& & $\widehat{\rho}_3^{or}$ & 0.004 & -0.011 & 0.178 & -0.009 & -0.013 & 0.186 & 0.001 & 0.009 & 0.100 & -0.012 & 0.010 & 0.132 \\
			\addlinespace[3pt]
			& \multirow{6}{*}{RISA-ADMM} 
			& $\widehat{\rho}_1$      & 4.041 & 4.024 & 0.213 & 4.053 & 3.975 & 0.220 & 4.015 & 4.021 & 0.183 & 4.026 & 4.007 & 0.202 \\
			& & $\widehat{\rho}_1^{or}$ & 3.970 & 3.960 & 0.188 & 3.967 & 3.966 & 0.199 & 4.011 & 4.011 & 0.142 & 4.013 & 4.012 & 0.150 \\
			& & $\widehat{\rho}_2$      & -4.027 & -4.029 & 0.188 & -4.034 & -4.002 & 0.194 & -4.008 & -4.000 & 0.176 & -4.012 & -4.013 & 0.181 \\
			& & $\widehat{\rho}_2^{or}$ & -3.995 & -4.008 & 0.157 & -3.990 & -4.008 & 0.176 & -4.003 & -4.000 & 0.142 & -4.005 & -4.005 & 0.154 \\
			& & $\widehat{\rho}_3$      & 0.048 & -0.001 & 0.189 & 0.053 & -0.018 & 0.051 & 0.000 & -0.011 & 0.104 & -0.020 & -0.016 & 0.152 \\
			& & $\widehat{\rho}_3^{or}$ & -0.008 & -0.020 & 0.164 & -0.029 & -0.056 & 0.183 & -0.005 & -0.006 & 0.086 & -0.008 & -0.008 & 0.150 \\
			\bottomrule
		\end{tabular}%
	}
\end{table*}

\newpage

To examine the estimation accuracy of $\bbeta$, Figure~\ref{fig: RMSE-Example1} presents boxplots of the root mean square errors (RMSEs) of 
the estimates $\widehat{\bbeta}$
obtained by BJ-ADMM and RISA-ADMM, where the RMSE of $\widehat{\bbeta}$ is defined as $\|\widehat{\bbeta}-\bbeta\|/\sqrt{p}$. 
As expected, larger sample sizes and lower censoring rates improve the estimation performance of $\bbeta$
for both BJ-ADMM and RISA-ADMM.
In Case~1 with normally distributed errors, 
both methods perform comparably, although BJ-ADMM attains slightly lower median RMSEs and variability, particularly when $n=100$.  In contrast, in Cases~2 and~3, which involve heteroscedastic errors or heavier-tailed distributions, 
RISA-ADMM consistently outperforms BJ-ADMM, with
lower median RMSEs, smaller variability, and fewer extreme outliers.

To further assess the accuracy of the subgroup-specific effect estimates $\widehat{\rho}_k$, we report the mean, median, and standard deviation of the estimates $\hat{\rho}_1$, $\hat{\rho}_2$, and $\hat{\rho}_3$ obtained using the BJ-ADMM and RISA-ADMM methods, along with their oracle counterparts, in Table~\ref{tab: rho-hat-Example1}. The oracle estimators assume known group memberships and differ in their treatment of censoring: BJ-ADMM uses imputation, whereas RISA-ADMM applies inverse probability weighting. 
In Case 1, in which errors follow a normal distribution, both BJ-ADMM and RISA-ADMM yield accurate and stable estimates with similar variability.  In Case 2 with  heteroscedastic errors, our method RISA-ADMM demonstrates improved robustness, producing estimates with lower variability and more consistent medians compared to BJ-ADMM. The difference 
becomes more evident for $t$-distributed errors (Case 3), where our method RISA-ADMM maintains better accuracy and reduced dispersion, while BJ-ADMM shows increased sensitivity to outliers and heavier tails.  Across all cases, oracle estimators, which assume known subgroup membership, achieve the lowest variability, confirming the benefits of accurate subgroup identification.

\subsection{Example 2 (Multiple heterogeneous effects)}

%
%

In our second example, we generate the failure times $T_i$ from a heterogeneous AFT model
\begin{align*}
	\log (T_i)=  z_i^{\top}\alpha_i + \bx_i^{\top}\bbeta + \varepsilon_i,\quad \,\, i=1, 2, \cdots, n,
\end{align*}
where $\bz=(z_1,z_2)^{\top}$ follows a bivariate standard normal distribution, each subgroup-specific coefficient vector
$\balpha_i$ is drawn from three different vectors $\brho_1=(4,4)^{\top}$, $\brho_2=(-4,-4)^{\top}$, and $\brho_3=(0,0)^{\top}$ with equal probability, and errors satisfy $\epsilon_i \stackrel{i.i.d.}{\sim} 0.5t_{(3)}$. 
Equivalently, 
there are three subgroups $G_1$, 
$G_2$, and $G_3$ and $\balpha_i=\brho_k$ for $i\in G_k$, $k=1, 2, 3$. 
Additionally, $\bx_i$, $\bbeta$,
and censoring time $C_i$ 
are generated as in Example 1.

\begin{table*}[h!]
	\centering
	\caption{The mean, median, and standard deviation (SD) of $\widehat{K}$, as well as the proportion of correctly identifying the number of subgroups (i.e., $P(\widehat{K} = K)$) for BJ-ADMM and RISA-ADMM over 100 replications in Example 2.}
	\label{tab: Khat-Example2}
	\resizebox{\textwidth}{!}{%
		\begin{tabular}{cc*{8}{c}} 
			\toprule
			& & \multicolumn{4}{c}{BJ-ADMM} & \multicolumn{4}{c}{RISA-ADMM} \\
			\cmidrule(lr){3-6} \cmidrule(lr){7-10}
			Sample size & Censoring & {Mean} & {Median} & {SD} & {$P(\widehat{K}=K)$} & {Mean} & {Median} & {SD} & {$P(\widehat{K}=K)$} \\
			\midrule
			\multirow{2}{*}{$n=100$} 
			& {20\%} & 3.84 & 4 & 0.86 & 0.45 & 3.27 & 3 & 0.65 & 0.77 \\
			& {40\%} & 4.00 & 4 & 1.08 & 0.42 & 3.39 & 3 & 0.72 & 0.69 \\
			\addlinespace
			\multirow{2}{*}{$n=200$} 
			& {20\%} & 3.69 & 3 & 0.72 & 0.53 & 3.23 & 3 & 0.55 & 0.80 \\
			& {40\%} & 3.81 & 3 & 0.80 & 0.48 & 3.32 & 3 & 0.58 & 0.75 \\
			\bottomrule
		\end{tabular}%
	}
\end{table*}

Table~\ref{tab: Khat-Example2} reports the mean, median, and standard deviation of the estimated number of subgroups $\widehat{K}$,  as well as the proportion of correctly identifying the 
number of subgroups (i.e., $P(\widehat{K} = K)$) 
for BJ-ADMM and RISA-ADMM under different settings in Example~2 based on 100 replications.  
It can be seen that our method
RISA-ADMM provides more accurate, stable, and robust estimation of the number of subgroups than BJ-ADMM across all settings in Example 2.
Specifically, RISA-ADMM yields mean and median estimates of the number of subgroups that are consistently closer to the true value, whereas BJ-ADMM tends to overestimate the number of subgroups, exhibiting larger means and standard deviations. 
In particular, RISA-ADMM consistently attains a median 
$\widehat{K}$ of 3 across all settings, whereas BJ-ADMM produces a median estimate of 4 under the setting $n=100$, even though the true number of subgroups is 3.
In addition, RISA-ADMM achieves substantially higher proportions of correct identification $P(\widehat{K}=K)$
than BJ-ADMM in each setting of Example 2.

Figure~\ref{fig: RMSE-Example2} presents boxplots of 
RMSEs of 
the estimates $\widehat{\bbeta}$ obtained by BJ-ADMM and RISA-ADMM. 
It shows that RISA-ADMM consistently outperforms BJ-ADMM in all settings, with
lower median RMSEs and smaller variability.
As expected, larger sample sizes and lower censoring rates improve the estimation performance of $\bbeta$
for both BJ-ADMM and RISA-ADMM.


\vspace{-2mm}
\begin{figure}[h]
	\centering
	\includegraphics[width=\linewidth]{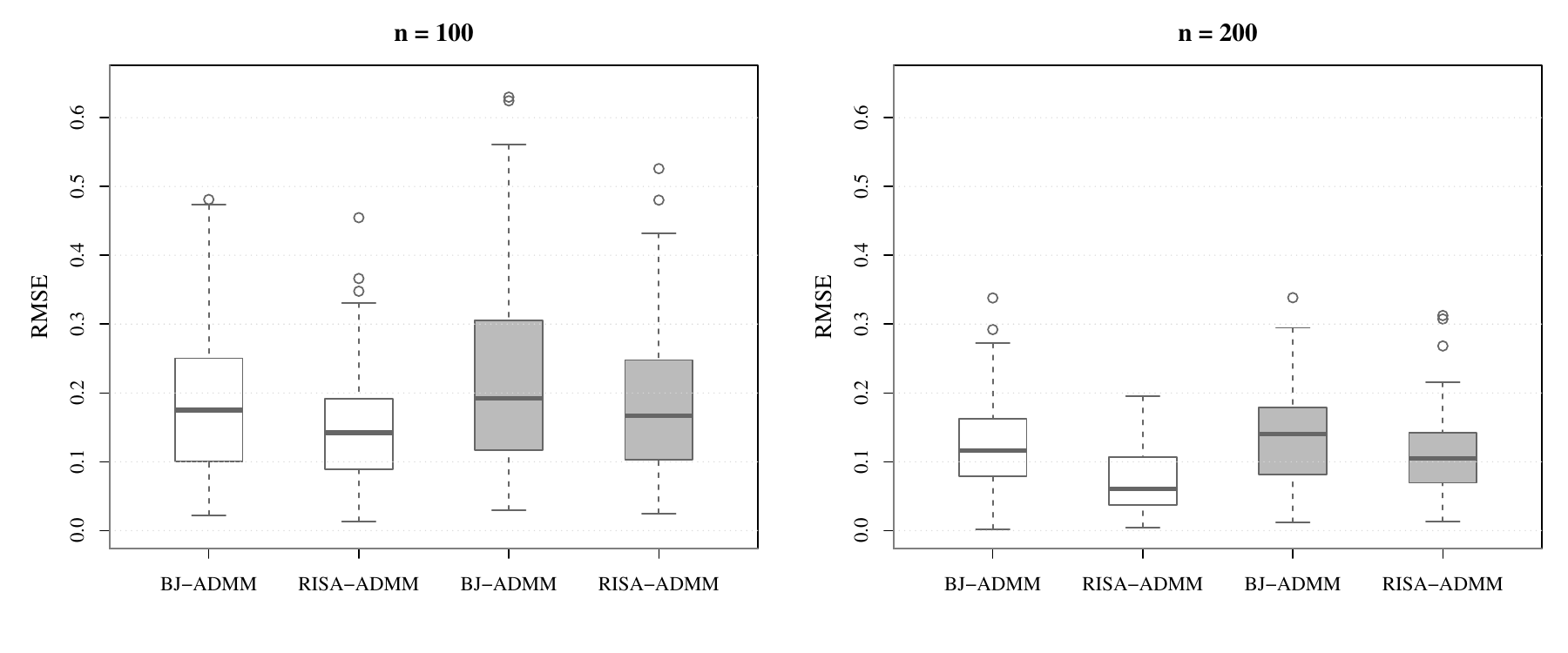}
	\caption{Boxplots of the RMSE of $\widehat{\bbeta}$ for BJ-ADMM and RISA-ADMM over 100 replications in Example 2. White: $20\%$ censoring; Grey: $40\%$ censoring.}
	\label{fig: RMSE-Example2}
\end{figure}

	\begin{table}[h]
		\centering
		\caption{The mean, median, and standard deviation (SD) of
			the root mean square errors (RMSEs) of $\widehat{\boldsymbol{\rho}}$ 
			in Example 2.}
		\label{tab: rho-RMSE-Example2}
			\scalebox{0.87}{
				\begin{tabular}{clcccccc}
					\toprule
					& & \multicolumn{3}{c}{Censoring = 20\%} & \multicolumn{3}{c}{Censoring = 40\%} \\
					\cmidrule(lr){3-5} \cmidrule(lr){6-8}
					$n$ & {Method} & {Mean} & {Median} & {SD} & {Mean} & {Median} & {SD} \\
					\midrule
					100 & BJ-ADMM & 0.385 & 0.334 & 0.217 & 0.601 & 0.411 & 0.758  \\
					& RISA-ADMM & 0.243 & 0.252 & 0.132 & 0.294 & 0.202 & 0.219  \\
					\addlinespace
					200 & BJ-ADMM  & 0.259 & 0.261 & 0.114 & 0.276 & 0.246 & 0.169 \\
					& RISA-ADMM & 0.158 & 0.145 & 0.061 & 0.202 & 0.197 & 0.074 \\
					\bottomrule
				\end{tabular}
			}
		\end{table}

		Table~\ref{tab: rho-RMSE-Example2} reports the mean, median and standard deviation of the RMSEs of $\widehat{\boldsymbol{\rho}}$ in Example~2, where the RMSE is defined as $\lVert \widehat{\boldsymbol{\rho}} - \boldsymbol{\rho} \rVert / \sqrt{Kq}$. 
		We observe that 
		RISA-ADMM consistently outperforms BJ-ADMM across all sample sizes and censoring rates, achieving lower mean and median RMSEs with smaller variability. The advantage is especially evident under heavier censoring.

		\section{Real data application: the German credit data}\label{section6}
		
		We apply our 
		proposed method to the German credit data, which was also analyzed by \cite{pei2024latent}.  
		This dataset classifies people described by a set of features as good or bad credit risks.
		The primary goal for banks is to understand who they are lending to. By analyzing subgroups defined by borrowers’ characteristics, banks can better identify how default risk varies across different subgroups. Once these differences are understood, borrowers can be categorized into distinct risk groups, enabling banks to make more informed decisions about whether to lend and under what terms.
		
		We define $T$ as 
		the default time and $z$ as a binary variable indicating the status of the checking account.  Our objective is to perform subgroup analysis under AFT models with $T$ as the response variable and $z$ as the treatment variable after adjusting for the effects of the covariates: credit amount $(x_1)$, present residence $(x_2)$,  installment rate $(x_3)$, purpose $(x_4)$, savings account $(x_5)$, other debtors or guarantors $(x_6)$, other installment plans $(x_7)$, housing $(x_8)$, and telephone $(x_9)$. 
		We randomly select 500 debtors from the original data for our analysis. Following~\cite{pei2024latent}, in our analysis, continuous variables are standardized to have mean zero and unit variance, and categorical variables 
		are encoded using dummy variables.
		To examine potential heterogeneity, we first fit a homogeneous AFT model using $T$ as the response and the
		nine covariates above and the treatment variable $z$ as predictors. 
		

		\begin{figure}[h]
			\centering
			\includegraphics[width=\linewidth]{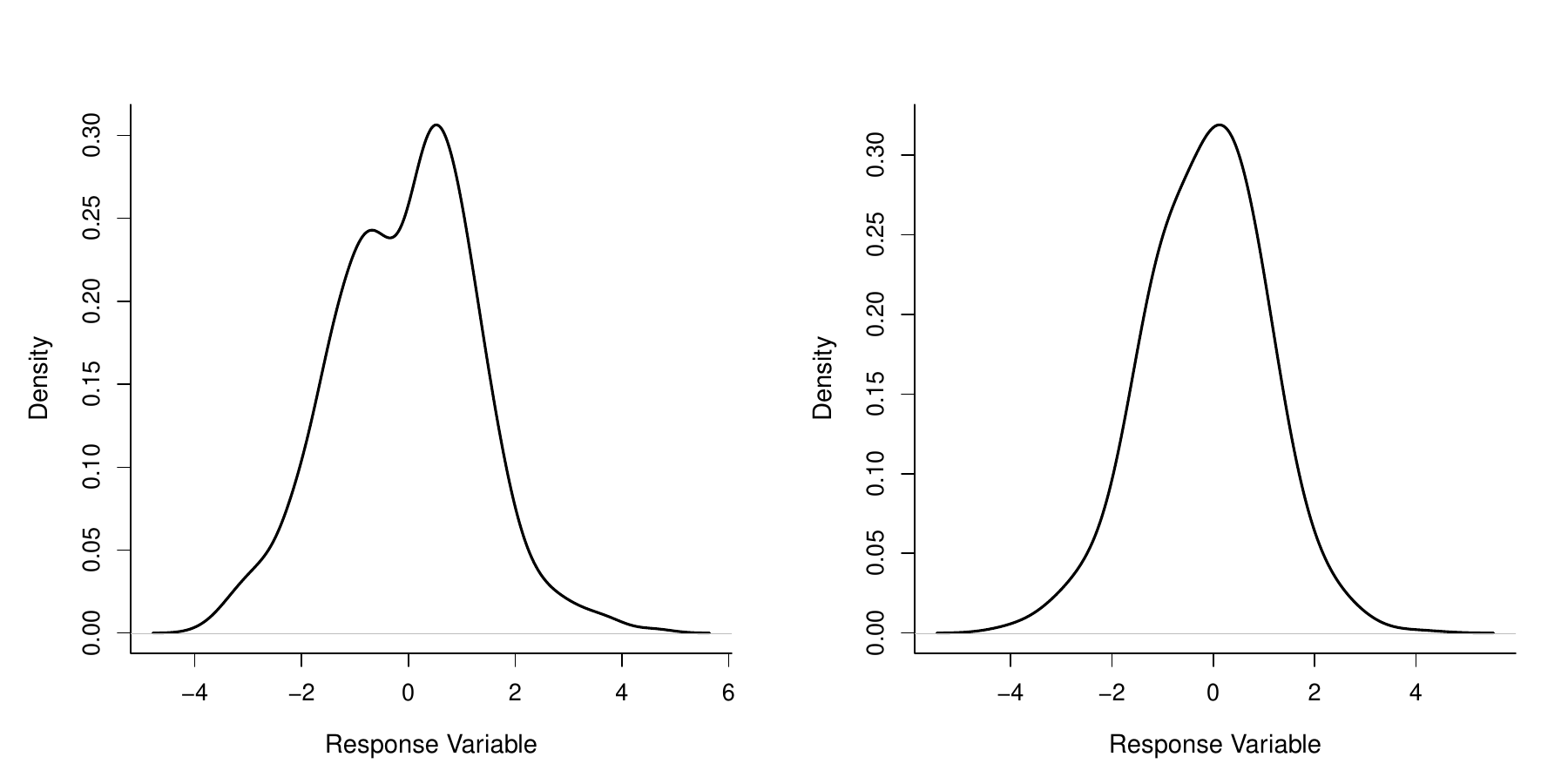}
			\caption{\normalsize 
					The kernel density plots of the residuals. Left: homogeneous AFT model; right: heterogeneous AFT model.}
			\label{realdataplot1}
		\end{figure}

		The kernel density plot of the resulting residuals is shown in the left panel of Figure~\ref{realdataplot1}. 
		The residual density plot clearly exhibits two modes, suggesting that heterogeneous effects may still be present even after adjusting for the covariates. 
		We then fit the heterogeneous AFT model $\log(T_i)=z_i\alpha_i+\bx_i^{\top}\bbeta+\varepsilon_i$, 
		where $T_i$, $z_i$, and $\bx_i$ are the copies of $T$, $z$, and $\bx=(x_1, \cdots, x_9)^{\top}$ for the $i$-th observation, respectively, and $\bbeta=(\beta_1, \cdots, \beta_9)^{\top}$. We detect potential heterogeneity using two methods: BJ-ADMM and RISA-ADMM, respectively.
		Notably, BJ-ADMM failed to detect any subgroups (i.e., $\widehat{K}=1$), whereas our proposed method, RISA-ADMM, identified $\widehat{K}=2$ subgroups. The finding obtained by our method RISA-ADMM is consistent with~\cite{pei2024latent}, who also reported $\widehat{K}=2$ subgroups. 
			The kernel density estimation of the residuals from our method, RISA-ADMM, is presented in the right panel of Figure 3. It can be seen that the residuals produced by our method exhibit a unimodal distribution. This pattern indicates that the proposed heterogeneous model better accommodates the latent subgroup structure and captures the underlying heterogeneity in the data.

		\begin{table*}[!ht]
			\centering
			\caption{The estimates, their standard errors (SE in parentheses), and $p$-values for testing parameter significance by BJ-ADMM and RISA-ADMM methods.}
			\label{realtable1}
			\begin{tabular*}{\textwidth}{@{\extracolsep{\fill}}lllll}
				\toprule
				& \multicolumn{2}{c}{BJ-ADMM} & \multicolumn{2}{c}{RISA-ADMM} \\
				\cmidrule(lr){2-3} \cmidrule(lr){4-5} 
				& Estimate (SE) & $p$-value & Estimate (SE) & $p$-value \\
				\midrule
				$\hat{\rho}_1$ (Subgroup 1)               & $5.552\times 10^{-5}$ ($1.493\times 10^{-5}$) & $<0.001$ & 1.125 (0.134) & $<0.001$ \\
				$\hat{\rho}_2$ (Subgroup 2)               & -- (--) & -- (--)      & -0.158 (0.093) & 0.089 \\
				$\hat{\beta}_1$ (Credit amount)           & 0.725 (0.060) & $<0.001$ & 0.557 (0.080) & $<0.001$ \\
				$\hat{\beta}_2$ (Present residence)       & 0.088 (0.051) & 0.085    & 0.083 (0.040) & 0.038 \\
				$\hat{\beta}_3$ (Installment rate)        & 0.602 (0.053) & $<0.001$ & 0.604 (0.081) & $<0.001$ \\
				$\hat{\beta}_4$ (Purpose)                 & -0.695 (0.106) & $<0.001$ & -1.412 (0.112) & $<0.001$ \\
				$\hat{\beta}_5$ (Saving account)          & -0.223 (0.108) & 0.039    & -1.191 (0.104) & $<0.001$ \\
				$\hat{\beta}_6$ (Other debtors)           & -0.112 (0.167) & 0.501    & -0.697 (0.165) & $<0.001$ \\
				$\hat{\beta}_7$ (Other installment plans) & 0.176 (0.125) & 0.160    & 0.362 (0.159) & 0.023 \\
				$\hat{\beta}_8$ (Housing)                 & 0.149 (0.085) & 0.082    & 0.145 (0.070) & 0.038 \\
				$\hat{\beta}_9$ (Telephone)               & 0.315 (0.110) & 0.004    & 0.979 (0.152) & $<0.001$ \\
				\bottomrule
			\end{tabular*}
		\end{table*}

		Table~\ref{realtable1} reports the estimated coefficients, their standard errors and the $p$-values for testing the significance of the coefficients by BJ-ADMM and RISA-ADMM.  
		As shown in Table~\ref{realtable1}, BJ-ADMM fails to identify variables present residence, other debtors, other installment plans, and housing as statistically significant at the $5\%$ level. By contrast, our method RISA-ADMM identifies all four variables as statistically significant at the same significance level.
		The results from RISA-ADMM are more reasonable. Intuitively, borrowers who have lived longer at their current residence tend to be more stable, a borrower’s housing situation signals financial stability, and borrowers with no other installment plan generally default later. Thus, variables present residence, other installment plans, and housing are associated with a longer time to default. On the other hand, the need for co-debtors may indicate weaker financial standing, as lenders often require co-debtors for higher-risk borrowers. In other words, the variable other debtors is associated with a shorter time to default. These findings suggest that accurately recovering heterogeneous structure help us
		identify important covariates that meaningfully affect the response.

		\section{Discussions}\label{section7}

		In this paper, we propose a new and robust subgroup analysis approach for 
		censored data under AFT models.  
		The proposed method accommodates censoring via inverse probability weighting and leverages concave pairwise fusion penalties within an M-estimation framework, allowing for automatic identification of latent subgroups without requiring prior knowledge of the underlying subgroup structure, while preserving high estimation accuracy. The effectiveness of the proposed method is validated by extensive simulation studies and a real data application.
		
			Although both covariate dimensions $p$ and $q$ are allowed to diverge with $n$ in our theoretical analysis, they remain smaller than the sample size $n$.
			It is interesting and challenging to extend our robust subgroup analysis approach to 
			the high-dimensional setting where
			both $p$ and $q$ can greatly exceed $n$. Our simulation study (Example 4) in the Supplementary Material shows that our method is insensitive to the robustification threshold $\tau$ in identifying the true number of subgroups, although the estimation accuracy of $\bbeta$ and $\brho$ may be affected by the choice of $\tau$. It would therefore be of interest to consider a data-driven choice of $\tau$ to further improve the estimation accuracy of the regression coefficients. 
			For example, one could investigate how to extend the idea of choosing a data-adaptive robustification threshold proposed in~\cite{sun2020Adaptive} to our subgroup analysis framework.
		These possible extensions
		are beyond the scope of the current paper and will be explored in future research.



\newpage

\setcounter{page}{1}
\setcounter{section}{0}

\setcounter{equation}{0}
\renewcommand{\thetable}{S\arabic{table}}
\renewcommand{\thefigure}{S\arabic{figure}}
\setcounter{table}{0}
\setcounter{figure}{0}



\begin{center}{\bf \Large Supplementary Material to ``Robust Subgroup Analysis for Heterogeneous Censored Data''}

	\bigskip
	
	Zhaohui Xu$^1$,  Daoji Li$^1$,  and Zemin Zheng$^1$
	
	$^1$ International Institute of Finance, School of
	Management, University of Science and
	Technology of China, Hefei, 230026, China\\
	
	$^2$Department of Information Systems and Decision Sciences, California State University, Fullerton, CA, 92831, United States

\end{center}


\setcounter{equation}{0}
\setcounter{table}{0}
\renewcommand{\theequation}{A.\arabic{equation}}
\renewcommand{\thetable}{S\arabic{table}}


This Supplementary Material consists of five parts. The first part presents the proof of Proposition~\ref{proposition1}. The second part provides the proofs of Theorems~\ref{theorem1} and~\ref{theorem2}. The third part reports results for homogeneous AFT models.
The fourth part examines the performance of the proposed method under different $\tau$
values and evaluates its sensitivity to the choice of $\tau$. The fifth part presents an additional simulation study with more groups and higher dimensions $p$ and $q$ to further assess the performance of the proposed method.

\section*{Appendix A: Proof of Proposition 1}\label{appendixA}

\renewcommand{\theequation}{A.\arabic{equation}}

First we define
\begin{align*}
	G^{(m+1)}=&\inf_{\footnotesize\bA\balpha^{(m+1)}-\bmeta=\bzero}\{\sum_{i=1}^n\frac{\Delta_i}{\widehat{S}(T_i^{*})}\rho\left(Y_i^{*}-\bz_i^{\top}\balpha_i-\bx_i^{\top}\bbeta\right)
	+ \sum_{1\le i< j\le n}p(\|\bet_{ij}\|, \lambda)\} \\
	=&\inf_{\footnotesize\bA\balpha^{(m+1)}-\bmeta=\bzero}L(\balpha^{(m+1)},\bbeta^{(m+1)},\bmeta,\bnu^{(m)}).
\end{align*}
It follows from the definition of $\bmeta^{(m+1)}$ in Section~\ref{section3} that
\begin{align*}
	L(\balpha^{(m+1)},\bbeta^{(m+1)},\bmeta^{(m+1)},\bnu^{(m)})\le	L(\balpha^{(m+1)},\bbeta^{(m+1)},\bmeta,\bnu^{(m)})
\end{align*}
for any $\bmeta$, and thus
\begin{eqnarray*}
	L(\balpha^{(m+1)},\bbeta^{(m+1)},\bmeta^{(m+1)},\bnu^{(m)})\le
	G^{(m+1)}.
\end{eqnarray*}	
Similarly, for any positive integer $t$, we can derive that
\begin{eqnarray*}
	L(\balpha^{(m+t)},\bbeta^{(m+t)},\bmeta^{(m+t)},\bnu^{(m+t-1)})\le
	G^{(m+t)}.
\end{eqnarray*}
Since the function $L(\balpha, \bbeta, \bmeta, \bnu)$ is differentiable with respect to $(\balpha,\bbeta)$ and is convex with respect to $\bmeta$, by applying the results in Theorem 4.1 of 
\cite{tseng2001convergence},
the sequence $(\balpha^{(m+1)},\bbeta^{(m+1)},\bmeta^{(m+1)})$ has a limit point, denoted by $(\balpha^*,\bbeta^*,\bmeta^*)$. Then we have
\begin{align*}
	G^*=&\lim_{m\to \infty}G^{(m+1)}=\lim_{m\to \infty}G^{(m+t)} \\
	=&
	\inf_{\footnotesize\bA\balpha^{*}-\bmeta^*=\bzero}\{\sum_{i=1}^n\frac{\Delta_i}{\widehat{S}(T_i^{*})}\rho\left(Y_i^{*}-\bz_i^{\top}\balpha_i^*-\bx_i^{\top}\bbeta^*\right)
	+ \sum_{1\le i< j\le n}p(\|\bmeta^*_{ij}\|, \lambda)\}.
\end{align*}
Moreover, by the definition of $\bnu$ in Section~\ref{section3} and some calculation, we have $\bnu^{(m+t-1)}=\bnu^{(m)}+\vartheta\sum_{i=1}^{t-1}(\bA\balpha^{(m+i)}-\bmeta^{(m+i)})$, thus
\begin{align*}
	&\quad	\lim_{m\to \infty}L(\balpha^{(m+t)},\bbeta^{(m+t)},\bmeta^{(m+t)},\bnu^{(m+t-1)})\\
	&=\sum_{i=1}^n\frac{\Delta_i}{\widehat{S}(T_i^{*})}\rho\left(Y_i^{*}-\bz_i^{\top}\balpha_i^*-\bx_i^{\top}\bbeta^*\right)
	+ \sum_{1\le i< j\le n}p(\|\bmeta^*_{ij}\|_2, \lambda) \\
	&\quad +\lim_{m\to \infty}\bnu^{(m)\top}(\bA\balpha^{*}-\bmeta^{*})
	+(t-\dfrac{1}{2})\vartheta\|\bA\balpha^{*}-\bmeta^{*}\|^2\\
	&\le G^{*}.
\end{align*}	
Therefore, $\lim_{m\to \infty}\|\br^{(m+1)}\|^2=\|\bA\balpha^{*}-\bmeta^{*}\|^2=0$.

\quad\quad Next, by the definition of $\balpha^{(m+1)}$, we have
\begin{align*}
	\partial L(\balpha^{(m+1)},\bbeta^{(m)},\bmeta^{(m)},\bnu^{(m)})/\partial \balpha=\bzero.
\end{align*}
Consequently,
\begin{align*}
	& \partial L(\balpha^{(m+1)},\bbeta^{(m)},\bmeta^{(m)},\bnu^{(m)})/\partial \balpha \\ =&\bZ^{\top}\brho(\balpha^{(m+1)},\bbeta^{(m)})+\bA^{\top}\bnu^{(m)}+\vartheta\bA^{\top}(\bA\balpha^{(m+1)}-\bmeta^{(m)})\\
	=&\bZ^{\top}\brho(\balpha^{(m+1)},\bbeta^{(m)})+\bA^{\top}(\bnu^{(m+1)}-\vartheta(\bA\balpha^{(m+1)}-\bmeta^{(m+1)})+\vartheta(\bA\balpha^{(m+1)}-\bmeta^{(m)}))\\
	=& \bZ^{\top}\brho(\balpha^{(m+1)},\bbeta^{(m)})+\bA^{\top}\bnu^{(m+1)}+\vartheta\bA^{\top}(\bmeta^{(m+1)}-\bmeta^{(m)}),
\end{align*}
where $\brho(\balpha^{(m+1)},\bbeta^{(m)})$ is a $n\times 1$ vector  and its $i$th element is $\frac{\Delta_i}{\widehat{S}(T_i^{*})}\dot{\rho}(Y_i^{*}-\bz_i^{\top}\balpha_i^{(m+1)}-\bx_i^{\top}\bbeta_i^{(m)})$. Therefore,
\begin{align*}
	\bs^{(m+1)}=\vartheta\bA^{\top}(\bmeta^{(m+1)}-\bmeta^{(m)})=-\bZ^{\top}\brho(\balpha^{(m+1)},\bbeta^{(m)})-\bA^{\top}\bnu^{(m+1)}.
\end{align*}
Since $\|\bA\balpha^{*}-\bmeta^{*}\|^2=0$, we have 
\begin{align*}
	&\quad\lim_{m\rightarrow \infty}\partial L(\balpha^{(m+1)},\bbeta^{(m)},\bmeta^{(m)},\bnu^{(m)})/\partial \balpha 
	=\lim_{m\rightarrow \infty}\bZ^{\top}\brho(\balpha^{(m+1)},\bbeta^{(m)})+\bA^{\top}\bnu^{(m+1)} \\
	=& \bZ^{\top}\brho(\balpha^*,\bbeta^*)+\bA^{\top}\bnu^*=\bzero.
\end{align*}
Thus, $\lim_{m\rightarrow \infty}\bs^{(m+1)}=\bzero$.

\phantomsection 
\label{appendixB}

\section*{Appendix B: Proofs of Theorems 1 and 2}\label{appendixB}

\setcounter{equation}{0}
\renewcommand{\theequation}{B.\arabic{equation}}

{\bf Proof of Theorem~\ref{theorem1}}. 
Define 
\begin{align*}
	\Phi_n(\bphi)=\sum_{i=1}^n\frac{\Delta_i}{\widehat{S}(T_i^{*})}\dot{\rho}(Y_i^{*}-\bU_i^{\top}\bphi)\bU_i \quad\mbox{and}\quad \Phi_n^S(\bphi)=\sum_{i=1}^n\frac{\Delta_i}{S(T_i^{*})}\dot{\rho}(Y_i^{*}-\bU_i^{\top}\bphi)\bU_i.
\end{align*}
Then by Condition 2, for every $r>0$, we have $\sup_{t<v}|\widehat{S}(t)-S(t)|=o(n^{-1/2+r})$, a.s.. This, together with Condition 4, implies that 
\begin{align}\label{th1-1}
	\sup_{\scriptsize{\bphi}}\|n^{-1}\Phi_n(\bphi)-n^{-1}\Phi_n^S(\bphi)\|=o(n^{-1/2+r}\sqrt{Kq+p}), \quad a.s..
\end{align}
Then under the condition
\begin{align*}
	\mathbb{P}(\lim_{n\rightarrow \infty} \{\inf_{\|\scriptsize{\bphi-\bphi_0}\|\ge n^{-\gamma}}\|\Phi_n^S(\bphi)\| \}/(n^{1/2+r}\sqrt{Kq+p})=\infty)=1, 
\end{align*}
we have
\begin{align*}
	\mathbb{P}(\Phi_n(\bphi)\ \mbox{have a zero point on}\ \|\bphi-\bphi_0\|\ge n^{-\gamma} \ \mbox{for large}\ n)=0.  
\end{align*}
Note that
\begin{align*}
	n^{-1}\Phi_n^S(\bphi)
	=& n^{-1}\Phi_n^S(\bphi_0)+n^{-1}\sum_{i=1}^n\frac{\Delta_i}{S(T_i^{*})}\ddot{\rho}(Y_i^{*}-\bU_i^{\top}\bphi_0)\bU_i\bU_i^{\top}(\bphi-\bphi_0) \\
	&\quad +o(n^{-1}\Lambda_{\max}(\bU^{\top}\bU)\|\bphi-\bphi_0\|).
\end{align*}
It follows from $\Phi_n(\widehat{\bphi}^{or})=0$ and the conditions $\Lambda_{\max}(\bU^{\top}\bU)\le c_4 n$ that  
\begin{align}\label{th1-2}
	\quad\sup_{\|\scriptsize{\widehat{\bphi}^{or}-\bphi_0}\|<n^{-\gamma}}\|n^{-1}\Phi_n^S(\bphi_0)+n^{-1}\bV_n(\widehat{\bphi}^{or}-\bphi_0)\|
	=o(\max\{n^{-\gamma},n^{-1/2+r}\sqrt{Kq+p}\}),
\end{align} 
where $\bV_n=\sum_{i=1}^n\frac{\Delta_i}{S(T_i^{*})}\rho''(Y_i^{*}-\bU_i^{\top}\bphi_0)\bU_i\bU_i^{\top}$. Since $\mathbb{E}[{\Phi_n^S(\bphi_0)}]=0$, and $\|V_n^{-1}\|\lesssim G_{\min}^{-1}$, we have
\begin{align*}
	\|\widehat{\bphi}^{or}-\bphi_0\|=o(\xi_n) \quad a.s.
\end{align*}
and
\begin{align*}
	\|\widehat{\brho}^{or}-\brho_0\|=\|\widehat{\bbeta}^{or}-\bbeta_0\|=o(\xi_n) \quad a.s.,
\end{align*}
where 
$\xi_n=\max\{n^{1-\gamma}/G_{\min},n^{1/2+r}\sqrt{Kq+p}/G_{\min}\}$. 
Moreover, 
\begin{align*}
	\|\widehat{\balpha}^{or}-\balpha\|^2=\sum_{k=1}^{K}\sum_{i \in G_k}\|\widehat{\brho}^{or}_{k}-\brho_{0k}\|^2\le G_{\max}\|_2\widehat{\brho}^{or}-\brho\|^2 
	=o(G_{\max}\xi^2_n) \quad a.s.,
\end{align*}
and
\begin{align*}
	\sup_{i}\|\widehat{\balpha}^{or}_{i}-\balpha_{0i}\|=\sup_{k}\|\widehat{\brho}^{or}_{k}-\brho_{0k}\|\le\|\widehat{\brho}^{or}-\brho_0\|=o(\xi_n) \quad a.s..
\end{align*}

Next, it follows from 
\eqref{th1-2}
that
\begin{align*}
	\widehat{\bphi}^{or}-\bphi_0=-\bV_n^{-1}\Phi_n^S(\bphi_0)+o(\xi_n) 
	=\sum_{i=1}^{n}\bV_n D_i(\bphi_0)+o(\xi_n),
\end{align*}
where $D_i(\bphi_0)=\dfrac{\Delta_i}{S(T_i^{*})}\rho'(Y_i^{*}-\bU_i^{\top}\bphi_0)\bU_i$. Then we verify the Lindeberg-Feller condition. It follows that
\begin{align*}
	\mathbb{E}\|\bV_n^{-1}D_i(\bphi_0)\|^4
	=& \mathbb{E}\{D_i(\bphi_0)^{\top}\bV_n^{-1}\bV_n^{-1}D_i(\bphi_0)\}^2 
	\le \|\bV_n^{-1}\|^4E\{D_i(\bphi_0)^{\top}D_i(\bphi_0)\}^2 \\
	=& O((Kq+p)^2/G_{\min}^{4}),
\end{align*}
and
\begin{align*}
	\mathbb{P}(\|\bV_n^{-1}D_i(\bphi_0)\|>\epsilon)\le \|\bV_n^{-1}\|^2E\|D_i(\bphi_0)\|^2/\epsilon^2
	=O((Kq+p)/(G_{\min}\epsilon)^{2})
\end{align*}
Thus, with the condition $\xi_n\rightarrow 0$, we have
\begin{align*}
	&\quad\sum_{i=1}^{n}\mathbb{E}\|\bV_n^{-1}D_i(\bphi_0)\|^2I(\|\bV_n^{-1}D_i(\bphi_0)\|>\epsilon)
	=n\mathbb{E}\|\bV_n^{-1}D_1(\bphi_0)\|^2I(\|\bV_n^{-1}D_1(\bphi_0)\|>\epsilon)\\
	\le & n\{\mathbb{E}\|\bV_n^{-1}D_1(\bphi_0)\|^4\}^{1/2}\{\mathbb{P}(\|\bV_n^{-1}D_1(\bphi_0)\|>\epsilon)\}^{1/2}
	=O(n(Kq+p)^{3/2}/G_{\min}^3)\rightarrow 0.
\end{align*}
Note that $\sum_{i=1}^{n}\var\{\bV_nD_i(\bphi_0)\}=\mathbb{E}(\bV_n^{-1}\bSigma_n\bV_n^{-1})$. Then it follows from the Lindeberg-Feller central limit theorem that
\begin{align*}
	\ba_n^{\top} \bmathcalV_n^{-1/2}(\widehat{\bphi}^{or}-\bphi_0)\rightarrow N(0,1)
\end{align*}
for any unit column vector $\ba_n\in\mathbb{R}^{p+Kq}$
as $n\to\infty$. 
This completes the proof of Theorem~\ref{theorem1}.


\vspace{2mm}

\noindent{\bf Proof of Theorem~\ref{theorem2}}.  The following proof is conditional on the event such that the results in Theorem~\ref{theorem1} hold. Before the proof of Theorem~\ref{theorem2}, we need to introduce a few additional notations used subsequently in the proof. Denote
\begin{align*}
	&F_n(\balpha,\bbeta)=\sum_{i=1}^n\frac{\Delta_i}{\widehat{S}(T_i^{*})}\rho(Y_i^{*}-\bz_i^{\top}\balpha_i-\bx_i^{\top}\bbeta),\\
	& L_n(\balpha)=\lambda\sum_{1\le i<j\le n}\varrho(\|\balpha_i-\balpha_j\|),\\
	&F_n^{\mathcal{G}}(\brho,\bbeta)=\sum_{i=1}^K\sum_{j\in G_i}\frac{\Delta_j}{\widehat{S}(T_j^{*})}\rho(Y_j^{*}-\bz_j^{\top}\brho_i-\bx_j^{\top}\bbeta),\\ &L_n^{\mathcal{G}}(\brho)=\lambda\sum_{1\le k< k'\le K}|G_{k}\|G_{k'}|\varrho(\|\brho_k-\brho_{k'}\|),
\end{align*}
and let $Q_n(\balpha,\bbeta)=F_n(\balpha,\bbeta)+L_n(\balpha)$, $Q_n^{\mathcal{G}}(\brho,\bbeta)=F_n^{\mathcal{G}}(\brho,\bbeta)+L_n^{\mathcal{G}}(\brho)$. We define the mappings $\Re:\Omega \to \mathbb{R}^{Kq}$ and $\Re^*:\mathbb{R}^{nq}\to \mathbb{R}^{Kq}$ such that $\Re(\bmu)$ is the $Kq\times1$ vector consisting of $K$ vectors with dimension $q$ and its $k^{th}$ vector equals to the common value of $\balpha_i$ for $i \in G_{k}$, and $\Re^*(\balpha)=\{|G_{k}|^{-1}\sum_{i\in G_{k}}\balpha_i^{\top},k=1,\dots,K\}^{\top}$. Denote the neighborhood of $(\balpha_0,\bbeta_0)$:
\begin{eqnarray*}
	\Theta=\{\balpha \in R^{nq},\bbeta \in R^p:\sup_i\|\balpha_i-\balpha_{0i}\|\le \xi_n,\|\bbeta-\bbeta_0\| \le \xi_n\},
\end{eqnarray*}	
where $\xi_n$ is given in Theorem~\ref{theorem1}.

By some simple calculation, for every $\balpha \in \Omega$, we have $\Re(\balpha)=\Re^*(\balpha)$, and $L_n(\balpha)=L_n^{\mathcal{G}}(\Re(\balpha))$. For every $\brho \in R^{Kq}$, we have $L_n(\Re^{-1}(\brho))=L_n^{\mathcal{G}}(\brho)$. Hence it is easy to obtain that
\begin{eqnarray}\label{th2-1}
	Q_n(\balpha,\bbeta)=Q_n^{\mathcal{G}}(\Re(\balpha),\bbeta),\quad Q_n^{\mathcal{G}}(\brho,\bbeta)=Q_n(\Re^{-1}(\brho),\bbeta).
\end{eqnarray}
From Theorem~\ref{theorem1}, we know that there exists an event $E_1$ and on the event $E_1$,
\begin{eqnarray*}
	\sup_i\|\widehat{\balpha}^{or}_i-\balpha_{0i}\|\le \xi_n,\|\widehat{\bbeta}^{or}-\bbeta_0\| \le \xi_n
\end{eqnarray*}
with $P(E_1)$ approaching 1 as $n\rightarrow \infty$, thus $((\widehat{\balpha}^{or})^{\top},(\widehat{\bbeta}^{or})^{\top})^{\top}\in \Theta$ on the event $E_1$. To ease readability, we will finish the proof in two steps to show that $((\widehat{\balpha}^{or})^{\top},(\widehat{\bbeta}^{or})^{\top})^{\top}$ is a strictly local minimizer of $Q_n(\balpha,\bbeta)$ with probability approaching one.

\medskip

\noindent \textbf{Step 1}. We show that on the event $E_1$, $Q_n(\balpha^*,\bbeta)>Q_n(\widehat{\balpha}^{or},\widehat{\bbeta}^{or})$ for any $((\balpha^*)^{\top},(\bbeta)^{\top})^{\top}\in \Theta$ and $((\balpha^{*})^{\top},(\bbeta)^{\top})^{\top}	\ne ((\widehat{\balpha}^{or})^{\top},(\widehat{\bbeta}^{or})^{\top})^{\top}$, where $\balpha^*=\Re^{-1}(\Re^*(\balpha))$. Let $\Re^*(\balpha)=\brho=(\brho_1^{\top},\dots,\brho_K^{\top})^{\top}$. Since
\begin{align*}
	&\sup_k\|\brho_k-\brho_{0k}\|^2=\sup_{k}\bigg|\bigg\|G_{k}|^{-1}\sum_{i\in G_{k}}\balpha_i-\brho_{0k}\bigg|\bigg|^2\\
	&=\sup_{k}\bigg|\bigg|\sum_{i\in G_{k}}(\balpha_i-\balpha_{0i})/|G_{k}|\bigg|\bigg|^2
	=\sup_{k}|G_k|^{-2}\bigg|\bigg|\sum_{i\in G_{k}}(\balpha_i-\balpha_{0i})\bigg|\bigg|^2\\
	&\le\sup_{k}|G_k|^{-1}\sum_{i\in G_{k}}\|(\balpha_i-\balpha_{0i})\|^2
	\le\sup_i\|\balpha_i-\balpha_{0i}\|^2\le \xi_n^2,
\end{align*}
then for all $k \ne k^{'}$, we have
\begin{eqnarray*}
	\|\brho_k-\brho_{k^{'}}\|\ge\|\brho_{0k}-\brho_{0k'}\|-2\sup_k\|\brho_k-\brho_{0k}\|\ge b_n-2\xi_n>c \lambda,
\end{eqnarray*}
where the last inequality follows from the assumption that $b_n>c\lambda\gg \xi_n$. Then by Condition 1, $\varrho(\|\brho_k-\brho_{k^{'}}|)$ is a constant, and we have $L^{\mathcal{G}}_n(\Re^*(\balpha))=C_n$ for some constant $C_n$ which does not depend $\balpha$. Hence $Q^{\mathcal{G}}_n(\Re^*(\balpha),\bbeta)=F^{\mathcal{G}}_n(\Re^*(\balpha),\bbeta)+C_n$ for all $(\balpha^{\top},\bbeta^{\top})^{\top} \in \Theta$.

Since $((\widehat{\brho}^{or})^{\top},(\widehat{\bbeta}^{or})^{\top})^{\top}$ is the unique global minimizer of $F^{\mathcal{G}}(\brho,\bbeta)$, then we have $F^{\mathcal{G}}(\Re^*(\balpha),\bbeta)>F^{\mathcal{G}}(\widehat{\brho}^{or},\widehat{\bbeta}^{or})$, thus $Q^\varsigma_n(\Re^*(\balpha),\bbeta)>Q^{\mathcal{G}}(\widehat{\brho}^{or},\widehat{\bbeta}^{or})$ for all $(\Re^{*}(\balpha)^{\top},\bbeta^{\top})^{\top} \ne ((\widehat{\brho}^{or})^{\top},(\widehat{\bbeta}^{or})^{\top})^{\top}$. By (B.2), we have $Q^{\mathcal{G}}(\Re^*(\balpha),\bbeta)=Q_n(\Re^{-1}(\Re^*(\balpha)),\bbeta)=Q_n(\balpha^*,\bbeta)$, and $Q^{\mathcal{G}}(\widehat{\brho}^{or},\widehat{\bbeta}^{or})=Q_n(\widehat{\balpha}^{or},\widehat{\bbeta}^{or})$. Therefore, $Q_n(\balpha^*,\bbeta)>Q_n(\widehat{\balpha}^{or},\widehat{\bbeta}^{or})$ for all $\balpha^* \ne \widehat{\balpha}^{or}$.
\medskip

\noindent \textbf{Step 2}. We show that there exists an event $E_2$ with $P(E_2)=1$ such that on $E_1 \cap E_2$, there is a neighborhood of $(\widehat{\balpha}^{or},\widehat{\bbeta}^{or})$, denoted by $\Theta_n$, we have  $Q_n(\balpha,\bbeta)\ge Q_n(\balpha^*,\bbeta)$ for any $(\balpha^{\top},\bbeta^{\top})^{\top} \in \Theta_n \cap \Theta$ for sufficiently large n.

For a positive sequence $t_n$, let $\Theta_n=\{\balpha_i:\sup_i\|\balpha_i-\widehat{\balpha}^{or}_i\|\le t_n\}$. By Taylor's expansion, we have
\begin{eqnarray}\label{th2-2}
	Q_n(\balpha,\bbeta)-Q_n(\balpha^*,\bbeta)=\Gamma_1+\Gamma_2
\end{eqnarray}
for $(\balpha^{\top},\bbeta^{\top})^{\top}\in \Theta_n \cap \Theta$, where
\begin{eqnarray*}
	\Gamma_1&=&-\sum_{i=1}^n\dfrac{\Delta_i}{\widehat{S}(T_i^*)}\dot{\rho}(Y_i^*-\bz_i^{\top}\balpha_i-\bx_i^{\top}\bbeta)\bz_i^{\top}(\balpha_i-\balpha_i^*),\\
	\Gamma_2&=&\sum_{i=1}^{n}\dfrac{\partial L_n(\balpha^m)}{\partial \balpha_i^{\top}}(\balpha_i-\balpha^*_i),
\end{eqnarray*}
and $\balpha^m=m\balpha+(1-m)\balpha^*$ for some constant $m\in(0,1)$. We establish and simplify the bounds for $\Gamma_1$ and $\Gamma_2$, respectively. 


First, following similar arguments in \cite{ma2020exploration}, it is easy to obtain
\begin{align}\label{th2-3}
	\Gamma_2\ge \sum_{k=1}^K\sum_{i,j\in G_k,i<j} \lambda\dot{\varrho}(4t_n)\|\balpha_i-\balpha_j\|. 
\end{align}
Next, for $\Gamma_1$, by setting 
\begin{align*}
	\bQ_i=\dfrac{\Delta_i}{\widehat{S}(T_i^*)}\dot{\rho}(Y_i^*-\bz_i^{\top}\balpha_i-\bx_i^{\top}\bbeta)\bz_i,  
\end{align*}
and $\bQ=(\bQ_1^{\top},\dots,\bQ_n^{\top})^{\top}$, we have
\begin{align*}
	\Gamma_1=-\sum_{k=1}^K\sum_{i,j\in G_k,i<j}\dfrac{(\bQ_j-\bQ_i)^{\top}(\balpha_j-\balpha_i)}{|G_k|}.
\end{align*}
It follows from the Conditions 2-4 that $\sup_i\|\bQ_i\|< c_2\sqrt{q}$. Thus we have 
\begin{align}\label{th2-4}
	\bigg|\dfrac{(\bQ_j-\bQ_i)(\balpha_j-\balpha_i)}{|G_k|}\bigg| \le 2G_{\min}^{-1}\sup_i\|\bQ_i\|\|\balpha_i-\balpha_j\|
	\le 2c_2G_{\min}^{-1}\sqrt{q}\|\balpha_i-\balpha_j\|.
\end{align}
Therefore, 
by \eqref{th2-2}-\eqref{th2-4},
we have
\begin{align*}
	Q_n(\balpha,\bbeta)-Q_n(\balpha^*,\bbeta)\ge \sum_{k=1}^K\sum_{i,j\in G_k,i<j}\{\lambda\dot{\varrho}(4t_n)-2c_2G_{\min}^{-1}\sqrt{q}\}\|\balpha_i-\balpha_j\|.
\end{align*}
Letting $t_n\to 0$ yields  $
\dot{\varrho}(4t_n)\to 1$. 
Since $\lambda\gg G_{\min}^{-1}\sqrt{q}$, we have $Q_n(\balpha,\bbeta)-Q_n(\balpha^*,\bbeta)\ge 0$ for sufficiently large $n$, which completes the proof of Theorem~\ref{theorem2}.

\section*{Appendix C: Homogeneous AFT models}

In this section, we consider the homogeneous AFT model, which is given by
\begin{align*}
	\log(T_i)=\bz_i^{\top}\brho+\bx_i^{\top}\bbeta+\varepsilon_i,\quad i=1,\dots,n.
\end{align*}
It can be seen that this model is a special case of model \eqref{eq: Heterogeneous-AFT} with $\balpha_1=\dots=\balpha_n=\brho$ and $K=1$. We also define the oracle estimator for $(\bbeta,\brho)$ as 
\begin{align*}
	\widehat{\bphi}^{or}=(\widehat{\bbeta}^{or{\top}},\widehat{\brho}^{or{\top}})^{\top}=\mathop{\arg\min}_{\footnotesize{\bphi\in \mathbb{R}^{p+q}}}\sum_{i=1}^n\frac{\Delta_i}{\widehat{S}(T_i^{*})}\rho(Y_i^{*}-\bU_i^{*{\top}}\bphi),
\end{align*}
where $\bU^{*}_i=(\bx_i^{\top},\bz_i^{\top})^{\top}$. Let $\widehat{\balpha}^{or}=(\widehat{\balpha}_1^{or{\top}},\dots,\widehat{\balpha}_n^{or{\top}})^{\top}$ with $\widehat{\balpha}^{or}_i=\widehat{\brho}^{or}$ for all $i$. Let $\bphi_0=(\bbeta_0^{\top},\brho_0^{\top})^{\top}$ be the true coefficient vector. Define
\begin{align*}
	\bV^*_n=\sum_{i=1}^n\frac{\Delta_i}{S(T_i^{*})}\ddot{\rho}(Y_i^{*}-{\bU^{*}_i}^{\top}\bphi_0)\bU_i^*{\bU^{*}_i}^{\top}\quad\mbox{and}\quad
	\bSigma^*_n=\sum_{i=1}^{n}\frac{\Delta_i}{S^2(T_i^{*})}\{\dot{\rho}(Y_i^{*}-{\bU^{*}_i}^{\top}\bphi_0)\}^2\bU_i^*{\bU^{*}_i}^{\top}.
\end{align*}
Finally, denote by $\bmathcalV_n^{*}=E(\bV_n^{*-1}\bSigma_n^{*}\bV_n^{*-1})$. Based on the composition of $\bU^{*}=(\bU_1^{*{\top}},\dots,\bU_n^{*{\top}})^{\top}$, we correspondingly decompose $\bmathcalV_n^{*}$ as
\begin{align*}
	\bmathcalV_n^{*}=\left(
	{ \begin{array}{cc}
			\bmathcalV_{n11}^{*}&\bmathcalV_{n12}^{*}\\
			\bmathcalV_{n21}^{*}&\bmathcalV_{n22}^{*}\\
	\end{array} }
	\right),
\end{align*}
where $ \bmathcalV_{n11}^{*}$ is a $p\times p$ matrix.  
To establish our theoretical results for homogeneous AFT models, we replace Condition 4 with Condition 5.

\begin{condition}\label{c4*}
	Assume $\sup_i\|\bx_i\|\le c_1\sqrt{p}$, $\sup_i\|\bz_i\|\le c_2 \sqrt{q}$, and $\Lambda_{\min}({\bU^*}^{\top}\bU^*)\ge c_3 n$, and $\Lambda_{\max}({\bU^*}^{\top}\bU^*)\le c_4 n$, for some constants $c_1, c_2, c_3,c_4\in(0,\infty)$. 
\end{condition}
Let $	\Phi_n^{*S}(\bphi)=\sum_{i=1}^n\left[\Delta_i/S(T_i^{*})\right]\dot{\rho}(Y_i^{*}-{\bU^{*}_i}^{\top}\bphi)\bU_i^*$.
Then we have the following results.


\begin{theo}\label{theorem3}
	Under Conditions 1-3 and 5, assume that
	\begin{align*}
		\mathbb{P}(\lim_{n\rightarrow \infty} \{\inf_{\|\footnotesize{\bphi-\bphi_0}\|\ge n^{-\gamma}}\|\Phi_n^{*S}(\bphi)\| \}/(n^{1/2+r}\sqrt{q+p})=\infty)=1,  
	\end{align*}
	with $0<r<1/2$ and $\gamma>0$, we have
	\begin{align*}
		&\|((\widehat{\bbeta}^{or}-\bbeta_0)^{\top},(\widehat{\brho}^{or}-\brho_0)^{\top})^{\top}\|=o(\xi'_{n})\quad a.s.,\notag\\
		&\sup_{i}\|\widehat{\balpha}^{or}_i-\balpha_{0i}\|=o(\xi'_n)\quad a.s., 
	\end{align*}
	where 
	$\xi'_{n}=\max\{n^{-\gamma},n^{-1/2+r}\sqrt{q+p}\}$. 
	And if $\xi'_n\rightarrow 0$, then for any unit column vector $\ba_n\in\mathbb{R}^{p+q}$, we have 
	\begin{align*}
		\ba_n^{\top} \bmathcalV_n^{*-1/2}(\widehat{\bphi}^{or}-\bphi_0)\xrightarrow{d} N(0,1).
	\end{align*}
	Moreover, if $b_n>c\lambda$ with 
	$\lambda\gg \xi_n'$, 
	there exists a local minimizer $(\widehat{\balpha}(\lambda),\widehat{\bbeta}(\lambda))$ of the object function $Q_n(\balpha,\bbeta,\lambda)$ given in (2.3) that satisfies
	\begin{align*}
		\mathbb{P}((\widehat{\balpha}(\lambda),\widehat{\bbeta}(\lambda))=(\widehat{\balpha}^{or},\widehat{\bbeta}^{or}))\rightarrow 1.
	\end{align*}
\end{theo}

\begin{coll}
	Under the conditions in Theorem~\ref{theorem3}, as $n\rightarrow \infty$, $\ba_n^{\top} \bmathcalV_n^{*-1/2}(\widehat{\bphi}-\bphi_0)\xrightarrow{d} N(0,1)$ for any unit column vector $\ba_n\in\mathbb{R}^{p+q}$. As a result, we have $\ba_{n1}^{\top}\bmathcalV_{n11}^{*-1/2}(\widehat{\bbeta}-\bbeta_0)\xrightarrow{d} N(0,1)$ and $\ba_{n2}^{\top}\bmathcalV_{n22}^{*-1/2}(\widehat{\brho}-\brho_0)\xrightarrow{d} N(0,1)$, where
	$\ba_{n1}\in \mathbb{R}^{p}$ and $\ba_{n2}\in \mathbb{R}^{q}$ are two unit column vectors.
\end{coll}

\subsection*{Example 3 (Homogeneous AFT model)}

In this example, we generate the failure times  $T_i$  from the folloiwing homogeneous AFT model
\begin{align*}
	\log(T_i)=z_i\alpha+\bx_i^{\top}\bbeta+\varepsilon_i,\quad i=1,\dots,n,
\end{align*}
where $z_i$, $\bx_i$ and $\bbeta$ are  set in the same way as in Example 1. We set $\alpha=2$ and $n=100$. 
We consider three different cases for generating $\varepsilon_i$:  Case 1 (normal), $\varepsilon_i\stackrel{i . i . d .}{\sim} \mathcal{N}(1,0.2^2)$; Case 2 (heteroscedastic normal), $\varepsilon_i=\Phi(x_{i1})\widetilde{\varepsilon}_{i}$ where  $\Phi(\cdot)$ is the cumulative distribution function of $\mathcal{N}(0,1)$ and $\widetilde{\varepsilon}_{i}\stackrel{i . i . d .}{\sim}N(0,1)$; Case 3 (t-distribution): $\varepsilon_i\stackrel{i . i . d .}{\sim} 0.5t_{(3)}$.
In addition, we generate the censoring time $C_i$ from $\min(c,U(0,c+2))$, where $U(a, b)$ denotes the uniform distribution 
on the interval $(a, b)$ and $c$ controls the censoring rate. 
We set $c$ to yield a censoring rate of approximately $20\%$ and $40\%$, respectively. 
For comparison, we also report the results obtained using the BJ-ADMM method. Tables \ref{tablec1} and \ref{tablec2} present the estimation results for $\widehat{K}$, $\widehat{\balpha}$, and $\widehat{\bbeta}$ based on 100 replications. 
Similar to Examples 1 and 2, it can be seen that our method
RISA-ADMM performs similarly to BJ-ADMM when the error distribution is normal (Case 1), but consistently outperforms BJ-ADMM under heteroscedastic and heavy-tailed error distributions (Cases 2 and 3). 

	%


\begin{table}[!ht]
	\centering
	\caption{The mean, median, and standard deviation (SD) of $\widehat{K}$, as well as the proportion of correctly identifying the 
		number of subgroups (i.e., $P(\widehat{K} = K)$) 
		for BJ-ADMM and RISA-ADMM over 100 replications 
		in Example 3.}
	\label{tablec1}
		\scalebox{0.9}{
			\begin{tabular}{ll*{8}{c}}
				\toprule
				& & \multicolumn{2}{c}{Case 1} & \multicolumn{2}{c}{Case 2} & \multicolumn{2}{c}{Case 3} \\
				\cmidrule(lr){3-4} \cmidrule(lr){5-6} \cmidrule(lr){7-8}
				& Censoring & 20\% & 40\% & 20\% & 40\% & 20\% & 40\% \\
				\midrule
				\multirow{4}{*}{BJ-ADMM} 
				& Mean & 1.03 & 1.06 & 1.26 & 1.29 & 1.63 & 1.76 \\
				& Median & 1 & 1 & 1 & 1 & 1 & 2 \\
				& SD & 0.22 & 0.28 & 0.44 & 0.46 & 0.63 & 0.79 \\
				& $P(\widehat{K}=K)$ & 0.98 & 0.95 & 0.74 & 0.71 & 0.55 & 0.44 \\
				\addlinespace
				
				\multirow{4}{*}{RISA-ADMM} 
				& Mean & 1.00 & 1.02 & 1.03 & 1.04 & 1.12 & 1.15 \\
				& Median & 1 & 1 & 1 & 1 & 1 & 1 \\
				& SD & 0 & 0.14 & 0.17 & 0.20 & 0.33 & 0.36 \\
				& $P(\widehat{K}=K)$ & 1.00 & 0.98 & 0.97 & 0.96 & 0.88 & 0.85 \\
				\bottomrule
		\end{tabular} }
	\end{table}

\begin{table*}[!ht]
	\centering
	\setlength{\tabcolsep}{4.5pt} 
	\footnotesize 
	
	\caption{The mean, median, and standard deviation (SD) of $\widehat{\alpha}$, $\widehat{\alpha}^{\mathrm{or}}$ , and the RMSE of $\widehat{\bbeta}$ for BJ-ADMM and RISA-ADMM over 100 replications in Example 3.}
	\label{tablec2}
	
	\begin{tabular}{lll*{9}{c}}
		\toprule
		& &  
		& \multicolumn{3}{c}{Case 1} 
		& \multicolumn{3}{c}{Case 2} 
		& \multicolumn{3}{c}{Case 3} \\
		\cmidrule(lr){4-6} \cmidrule(lr){7-9} \cmidrule(lr){10-12}
		Method & Estimator & Censoring  & Mean & Median & SD & Mean & Median & SD & Mean & Median & SD \\
		\midrule
		
		\multirow{6}{*}{BJ-ADMM}
		& $\widehat{\alpha}$ & 20\% 
		& 2.005 & 2.005 & 0.023 
		& 2.017 & 2.029 & 0.084 
		& 1.985 & 1.982 & 0.116 \\
		&  & 40\% 
		& 2.006 & 2.006 & 0.026 
		& 2.020 & 2.027 & 0.087 
		& 1.977 & 1.980 & 0.130 \\
		
		& $\widehat{\alpha}^{\mathrm{or}}$ & 20\% 
		& 2.003 & 2.005 & 0.020 
		& 2.012 & 2.015 & 0.081 
		& 1.997 & 1.996 & 0.100 \\
		&  & 40\% 
		& 2.004 & 2.005 & 0.023 
		& 2.016 & 2.015 & 0.080 
		& 1.985 & 1.988 & 0.122 \\
		
		& $\|\widehat{\bbeta}-\bbeta_0\|/\sqrt{p}$ & 20\% 
		& 0.021 & 0.020 & 0.009 
		& 0.063 & 0.060 & 0.034 
		& 0.078 & 0.071 & 0.043 \\
		&  & 40\% 
		& 0.023 & 0.022 & 0.012 
		& 0.073 & 0.071 & 0.040 
		& 0.095 & 0.091 & 0.059 \\
		
		\midrule
		
		\multirow{6}{*}{RISA-ADMM}
		& $\widehat{\alpha}$ & 20\% 
		& 2.005 & 2.009 & 0.026 
		& 2.016 & 2.018 & 0.073 
		& 2.013 & 1.996 & 0.105 \\
		&  & 40\% 
		& 2.006 & 2.008 & 0.029 
		& 2.018 & 2.020 & 0.077 
		& 1.975 & 2.002 & 0.110 \\
		
		& $\widehat{\alpha}^{\mathrm{or}}$ & 20\% 
		& 2.007 & 2.009 & 0.026 
		& 2.016 & 2.013 & 0.070 
		& 1.990 & 1.995 & 0.105 \\
		&  & 40\% 
		& 2.006 & 2.008 & 0.028 
		& 2.018 & 2.020 & 0.075 
		& 1.997 & 1.999 & 0.106 \\
		
		& $\|\widehat{\bbeta}-\bbeta_0\|/\sqrt{p}$ & 20\% 
		& 0.022 & 0.022 & 0.011 
		& 0.055 & 0.061 & 0.031 
		& 0.068 & 0.082 & 0.035\\
		&  & 40\% 
		& 0.026 & 0.025 & 0.014 
		& 0.069 & 0.066 & 0.035 
		& 0.090 & 0.084 & 0.051 \\
		\bottomrule
	\end{tabular} 
\end{table*}

		
			\section*{Appendix D: Example 4 (Sensitivity analysis of $\tau$)}

			In the previous numerical studies, 
			following
			~\citep{huber1981robust,zhao2020statistical}, we set the robustification threshold $\tau$ in the Huber loss to 
			the default value $1.345$, i.e., 
			$\tau = \tau_0 = 1.345$, for simplicity.   
			As suggested by one referee, we investigate the sensitivity of the proposed method to the choice of  $\tau$ by examining its finite-sample numerical performance in estimating the number of subgroups $K$, the regression coefficient vector $\bbeta$, and the subject-specific effects $\brho$ across different values of $\tau$. To this end, we consider an additional simulation example (hereafter referred to as Example 4) in which
			errors satisfy $\epsilon_i \stackrel{i.i.d.}{\sim} 0.5t_{(3)}$ and all settings are identical to those in Example 2. We evaluate our RISA-ADMM method for $\tau = 0.5\tau_0$,  $\tau_0$, and $1.5\tau_0$, respectively.
			The results are presented in Table~\ref{table_tau}, Figure~\ref{fig: RMSE-Example4}, and Table~\ref{table_tau_rmse}.
			
			It can be seen from these results that our RISA-ADMM method is largely insensitive to the robustification threshold $\tau$ when identifying the true number of subgroups, although the estimation accuracy of $\bbeta$ and $\brho$ may be affected by the choice of $\tau$.  In addition, the sensitivity of the estimation accuracy of $\bbeta$ and $\brho$ to the choice of $\tau$ decreases as the sample size $n$ increases. 
			These findings indicate a clear distinction between subgroup identification and parameter estimation accuracy.
			Specifically, as shown in Table~\ref{table_tau}, the median estimated number of subgroups is consistently equal to the true number of subgroups, $3$, for all tested $\tau$ values, censoring levels, and sample sizes, 
			perfectly recovering the true group structure.  Moreover, the mean estimated number of subgroups shows only negligible differences across all settings. 
			However, as illustrated in Figure~\ref{fig: RMSE-Example4} and Table~\ref{table_tau_rmse}, the estimation accuracy of $\bbeta$ and $\brho$ 
			is more sensitive to the choice of $\tau$.
			In particular, when censoring rate is high ($40\%$), the RMSE of both $\widehat{\boldsymbol{\beta}}$ and $\widehat{\boldsymbol{\rho}}$ shows an upward trend as $\tau$ increases from $0.5\tau_0$ to $1.5\tau_0$. This is consistent with the properties of the Huber loss function: a larger $\tau$ increases the range of residuals treated quadratically, which improves efficiency under normal errors but reduces the ``robustness'' against outliers or heavy-tailed distributions in the censored data. 
			Notably, the sensitivity of the estimation accuracy of $\bbeta$ and $\brho$ to $\tau$ is significantly reduced as the sample size increases from $n=100$ to $n=200$.
			
			Although our RISA-ADMM method is insensitive to the robustification threshold $\tau$ 
			for identifying the number of subgroups, it would be of interest to consider a data-driven choice of $\tau$ to further improve the estimation accuracy of the regression coefficients. 
			For example, one could investigate how to extend the idea of choosing a data-adaptive robustification threshold proposed in~\cite{sun2020Adaptive} to our subgroup analysis framework.
			Since the current paper is already lengthy, we leave this topic for future research.

		\begin{table}[!htbp]
			\centering
			\setlength{\tabcolsep}{3.2mm}
			\renewcommand{\arraystretch}{1.15}
			\small
			
			\caption{The mean, median, and standard deviation (SD) of $\widehat{K}$, 
				as well as the proportion of correctly identifying the number of subgroups 
				(i.e., $P(\widehat{K}=K)$) for RISA-ADMM over 100 replications under different $\tau$ values in Example~4.}
			\label{table_tau}
				\scalebox{0.9}{\begin{tabular}{ll*{6}{c}}
						\toprule
						&  & \multicolumn{2}{c}{$0.5\tau_0$} & \multicolumn{2}{c}{$\tau_0$} & \multicolumn{2}{c}{$1.5\tau_0$} \\
						\cmidrule(lr){3-4}\cmidrule(lr){5-6}\cmidrule(lr){7-8}
						& Censoring & 20\% & 40\% & 20\% & 40\% & 20\% & 40\% \\
						\midrule
						\multirow{4}{*}{$n=100$}
						& Mean               & 3.25 & 3.30 & 3.27 & 3.39 & 3.29 & 3.36 \\
						& Median             & 3    & 3    & 3    & 3    & 3    & 3    \\
						& SD                 & 0.52 & 0.64 & 0.65 & 0.72 & 0.69 & 0.63 \\
						& $P(\widehat{K}=K)$ & 0.79 & 0.76 & 0.77 & 0.69 & 0.73 & 0.70 \\
						\addlinespace[4pt]
						\multirow{4}{*}{$n=200$}
						& Mean               & 3.22 & 3.27 & 3.23 & 3.32 & 3.24 & 3.33 \\
						& Median             & 3    & 3    & 3    & 3    & 3   & 3   \\
						& SD                 & 0.48 & 0.55 & 0.55 & 0.58 & 0.59 & 0.60 \\
						& $P(\widehat{K}=K)$ & 0.86 & 0.83 & 0.80 & 0.75 & 0.77 & 0.73 \\
						\bottomrule
				\end{tabular}}
		\end{table}

		\begin{figure}[!htbp]
			\centering
			\includegraphics[width=0.95\linewidth]{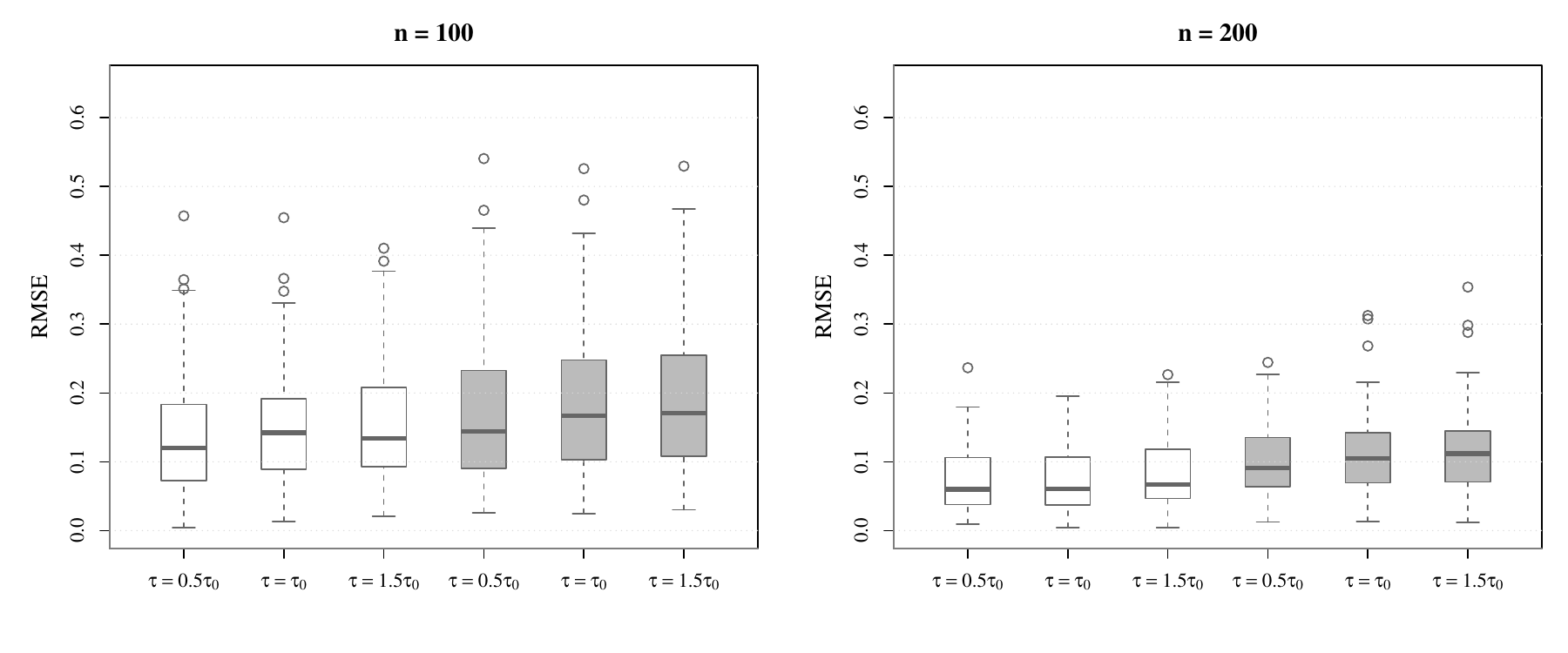}
			\caption{
					Boxplots of the RMSE of $\widehat{\bbeta}$ for RISA-ADMM over 100 replications under different $\tau$ values in Example 4. White: $20\%$ censoring; Grey: $40\%$ censoring.}
			\label{fig: RMSE-Example4}
		\end{figure}

		\begin{table}[!htbp]
			\centering
			\setlength{\tabcolsep}{3.2mm}
			\renewcommand{\arraystretch}{1.15}
			\small
			
			\caption{
					The mean, median, and standard deviation (SD) of the RMSE of $\widehat{\brho}$ for RISA-ADMM over 100 replications under different $\tau$ values in Example~4.}
			\label{table_tau_rmse}
				\scalebox{0.98}{\begin{tabular}{ll*{6}{c}}
						\toprule
						&  & \multicolumn{2}{c}{$0.5\tau_0$} & \multicolumn{2}{c}{$\tau_0$} & \multicolumn{2}{c}{$1.5\tau_0$} \\
						\cmidrule(lr){3-4}\cmidrule(lr){5-6}\cmidrule(lr){7-8}
						& Censoring & 20\% & 40\% & 20\% & 40\% & 20\% & 40\% \\
						\midrule
						
						\multirow{3}{*}{$n=100$}
						& Mean   & 0.271    & 0.281 & 0.243    & 0.294 & 0.263     & 0.302 \\
						& Median & 0.236    & 0.203 & 0.252    & 0.202 & 0.226     & 0.301 \\
						& SD     & 0.198    & 0.205 & 0.132   & 0.219 &  0.208    & 0.230 \\
						
						\addlinespace[4pt]
						
						\multirow{3}{*}{$n=200$}
						& Mean   &  0.160    & 0.173 & 0.158    & 0.202 &  0.175    & 0.221 \\
						& Median &  0.148    & 0.168 & 0.145     & 0.197 & 0.166     & 0.219 \\
						& SD     &  0.070    & 0.061 & 0.061     & 0.074 & 0.066     & 0.077 \\
						
						\bottomrule
				\end{tabular}}
		\end{table}

		\vspace{4mm}
		
			\section*{Appendix E: Example 5 (simulation with more groups and higher dimensions \protect\lowercase{$p$} and \protect\lowercase{$q$})}

			Following the suggestion of one referee, we conduct a simulation study (hereafter referred to as Example 5) with more groups and higher dimensions $p$ and $q$  to further assess the performance of the proposed method. 
			We compare our RISA-ADMM method with BJ-ADMM, the imputation-based approach proposed in~\cite{yan2021subgroup}.
			Specifically, we set $p = 20$, $q = 3$, $K=5$ and $n=300$. The  parameter vector $\bbeta \in \mathbb{R}^p$ is set such that its first $p/2$ components are equal to 1, while its remaining $p/2$ components are equal to $-1$. 
			The subgroup-specific coefficient vector $\brho_i$ is drawn with equal probability from the following five vectors:
			$\brho_1 = (4,4,4)^\top$, 
			$\brho_2 = (2,2,2)^\top$,
			$\brho_3 = (0,0,0)^\top$, $\brho_4=(-2,-2,-2)^{\top}$, and  $\brho_5=(-4,-4,-4)^{\top}$.
			All other simulation settings are the same as those in Example~2, with a censoring rate of $40\%$.
			We compare our RISA-ADMM method with BJ-ADMM. 
			Table \ref{example5-table1} reports the mean, median, and standard deviation of the estimated number of subgroups $\widehat{K}$, as well as the proportion of correctly identifying the true number of subgroups, $\mathbb{P}(\widehat{K}=K)$, based on 100 replications. The estimates of $\bbeta$ and $\brho$ are summarized in Figure \ref{fig: RMSE-Example5} and Table \ref{example5-table2}. 

		\begin{table}[!htbp]
			\centering
			\setlength{\tabcolsep}{3.2mm}
			\renewcommand{\arraystretch}{1.15}
			\small
			\caption{
					The mean, median, and standard deviation (SD) of $\widehat{K}$, as well as the proportion of correctly identifying the number of subgroups (i.e., $P(\widehat{K}=K)$) for BJ-ADMM and RISA-ADMM over 100 replications 
					in Example 5.}
			\label{example5-table1}
			\begin{tabular}{lcccc}
				\hline
				Method & Mean & Median & SD & $P(\widehat{K}=K)$ \\
				\hline
				BJ-ADMM   & 6.42 & 6 & 0.84 & 0.23 \\
				RISA-ADMM & 5.29 & 5 & 0.67 & 0.70 \\
				\hline
			\end{tabular}
		\end{table}

		\begin{figure}[!htbp]
			\centering
			\includegraphics[width=0.43\linewidth]{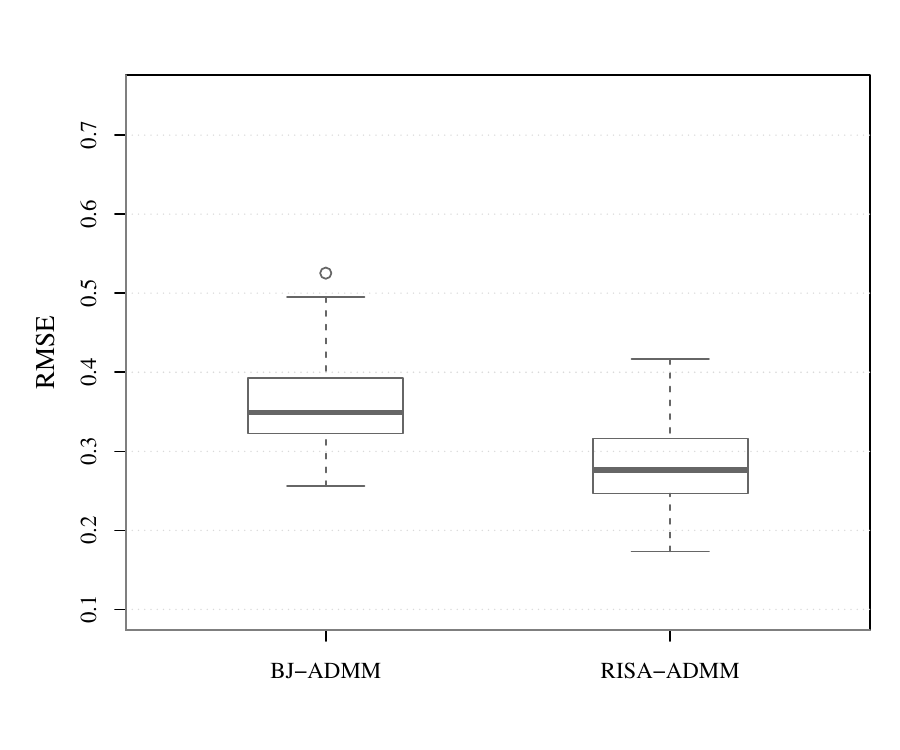}
			\caption{
					Boxplots of the RMSE of $\widehat{\bbeta}$ for BJ-ADMM and RISA-ADMM over 100 replications 
					in Example 5.}
			\label{fig: RMSE-Example5}
		\end{figure}
		
		\begin{table}
			\centering
			\setlength{\tabcolsep}{3.2mm}
			\renewcommand{\arraystretch}{1.15}
			\small
			\caption{
					The mean, median, and standard deviation (SD) of the RMSEs of $\widehat{\brho}$ for BJ-ADMM and RISA-ADMM 
					in Example 5.}
			\label{example5-table2}
			\begin{tabular}{lccc}
				\hline
				Method & Mean & Median & SD \\
				\hline
				BJ-ADMM   & 1.011 & 1.059 & 0.238 \\
				RISA-ADMM & 0.561 & 0.627 & 0.211 \\
				\hline
			\end{tabular}
		\end{table}

			These results further confirm that our method, RISA-ADMM, consistently outperforms BJ-ADMM. 
			In particular, Table \ref{example5-table1} shows that
			RISA-ADMM provides more accurate, stable, and robust estimates of the number of subgroups in this setting. Its mean and median estimates of $\widehat{K}$ are consistently closer to the true value, whereas BJ-ADMM tends to overestimate the number of subgroups and exhibits greater variability. Notably, RISA-ADMM attains a median $\widehat{K} = 5$, matching the true number of subgroups, while BJ-ADMM produces a median estimate of 6 despite the true value being 5. Furthermore, RISA-ADMM achieves substantially higher proportions of correct identification, $\mathbb{P}(\widehat{K} = K)$, than BJ-ADMM, demonstrating its robustness and effectiveness in higher-dimensional settings. In addition, both Figure \ref{fig: RMSE-Example5} and Table \ref{example5-table2} show that RISA-ADMM yields lower mean and median RMSEs than BJ-ADMM for estimating $\widehat{\bbeta}$ and $\widehat{\brho}$.


%
%
%
%

\end{document}